\documentclass[submission,copyright,creativecommons]{eptcs}
\usepackage{etex}
\providecommand{\event}{QPL 2019} 
\usepackage{breakurl}             
\usepackage{underscore}           

\usepackage{enumerate,xspace}
\usepackage{amsmath,amssymb,wasysym}
\usepackage[all]{xy}
\usepackage{stmaryrd} 
\usepackage{multicol}
\usepackage{proof}
\usepackage{cmll}

\usepackage{tikz}
\usetikzlibrary{decorations.markings,arrows,arrows.meta,positioning,calc}
\pgfdeclarelayer{edgelayer}
\pgfdeclarelayer{nodelayer}
\pgfsetlayers{edgelayer,nodelayer,main}

\tikzstyle{none}=[inner sep=0pt]
\tikzstyle{port}=[inner sep=0pt]
\tikzstyle{component}=[circle, fill=white, draw=black, inner sep=1pt]
\tikzstyle{integral}=[inner sep=0pt]
\tikzstyle{differential}=[inner sep=0pt]
\tikzstyle{codifferential}=[inner sep=0pt]
\tikzstyle{function1}=[fill=white, draw=black, minimum size=1cm]
\tikzstyle{function2}=[fill=white, draw=black, minimum size=0.5cm]
\tikzstyle{function3}=[fill=white, draw=black, minimum size=1.5cm]
\tikzstyle{duplicate}=[inner sep=1pt]
\tikzstyle{object}=[inner sep=2pt]

\tikzstyle{wire}=[-, draw=black, line width=1.000]
\tikzstyle{pre}=[<-, shorten <=1pt, >=stealth, draw=black, line width=1.000]
\tikzstyle{post}=[->, shorten >=1pt, >=stealth, draw=black, line width=1.000]

\tikzset{
  myblock/.style={
    draw,minimum size=1cm,align=center
  },
  twoarrows/.style n args={4}{
    decoration={
      markings,
      mark=at position #1 with {\arrow{>}\node[above] {#3};}, 
      mark=at position #2 with {\arrow{>}\node[above] {#4};} 
    },
  postaction=decorate
  },
  onearrow/.style 2 args={
    decoration={
      markings,
      mark=at position #1 with {\arrow[scale=1,>=stealth]{>}\node[right] {#2};}, 
    },
  postaction=decorate,
    line width=1.000
  },
   zeroarrow/.style 2 args={
    decoration={
      markings,
      mark=at position #1 with {\node[left] {#2};}, 
    },
  postaction=decorate,
  line width=1.000
  },
   mysquare/.style={
    draw,minimum size=0.5cm,align=center},
}

\newtheorem{observation}{Remark}[section]
\newtheorem{theorem}[observation]{Theorem}
\newtheorem{definition}[observation]{Definition}
\newtheorem{example}[observation]{Example}

\newtheorem{proposition}[observation]{Proposition} 
\newtheorem{corollary}[observation]{Corollary}

\title{Why $\mathsf{FHilb}$ is Not an Interesting (Co)Differential Category}
\author{Jean-Simon Pacaud Lemay 
\institute{Department of Computer Science \\University of Oxford \\Oxford, UK}
\email{jean-simon.lemay@kellogg.ox.ac.uk}
\email{https://www.cs.ox.ac.uk/people/jean-simon.lemay/}}
\def\titlerunning{Why $\mathsf{FHilb}$ is Not an Interesting (Co)Differential Category}
\def\authorrunning{J-S. P. Lemay}
\begin{document}
\allowdisplaybreaks
\maketitle

\begin{abstract}
Differential categories provide an axiomatization of the basics of differentiation and categorical models of differential linear logic. As differentiation is an important tool throughout quantum mechanics and quantum information, it makes sense to study applications of the theory of differential categories to categorical quantum foundations. In categorical quantum foundations, compact closed categories (and therefore traced symmetric monoidal categories) are one of the main objects of study, in particular the category of finite-dimensional Hilbert spaces $\mathsf{FHilb}$. In this paper, we will explain why the only differential category structure on FHilb is the trivial one. This follows from a sort of incompatibility between the trace of $\mathsf{FHilb}$ and possible differential category structure. That said, there are interesting non-trivial examples of traced/compact closed differential categories, which we also discuss. 

The goal of this paper is to introduce differential categories to the broader categorical quantum foundation community and hopefully open the door to further work in combining these two fields. While the main result of this paper may seem somewhat ``negative''  in achieving this goal, we discuss interesting potential applications of differential categories to categorical quantum foundations. 
\end{abstract}

\subparagraph*{Acknowledgements:} The author would like to thank Jamie Vicary and the audience of CT Oktoberfest 2018 for useful discussions, as well as Robert Seely for editorial comments. The author would also like to thank Kellogg College, the Clarendon Fund, and the Oxford-Google Deep Mind Graduate Scholarship for financial support.

\section{Introduction}

Differential categories were introduced by Blute, Cockett, and Seely \cite{blute2006differential} to provide categorical models of differential linear logic ($\mathsf{DILL})$ \cite{ehrhard2017introduction}, which is an extension of the multiplicative and exponential fragments of intuitionistic linear logic ($\mathsf{MELL}$) by adding a differentiation inference rule. The axioms of a differential category encode the basic properties of the derivative from differential calculus such as the product rule (also known as the Leibniz rule) and the chain rule. The dual notion of differential categories, codifferential categories, provides generalizations of derivation operators in commutative algebra \cite{blute2015derivations}. The theory of differential categories now has a rich literature of its own and has led to other abstract formulations of several notions of differentiation such as the directional derivative \cite{blute2009cartesian} and smooth manifolds \cite{cockett2014differential}. As differentiation is an important tool throughout quantum mechanics and quantum information, it makes sense to study applications of the theory of differential categories to categorical quantum foundations (as suggested briefly in the conclusion of \cite{laird2013constructing}). 

There are many models of $\mathsf{MELL}$ with a ``quantum foundation'' flavour in the literature, most of which are related to the bosonic Fock space \cite[Chapter 21]{geroch1985mathematical} such as studied by Fiore in \cite{fiore2015axiomatics}, Vicary in \cite{vicary2008categorical}, and Blute, Panangaden, and Seely in \cite{blute1994fock}. In fact, the models provided in these sources are all differential categories (though not necessarily studied as such), and therefore already provide perfectly good examples of differential categories in categorical quantum foundations. However, arguably the most fundamental objects of study in categorical quantum foundations are compact closed categories (which the models in \cite{vicary2008categorical} and \cite{blute1994fock} are not) and in particular $\mathsf{FHilb}$, the category of finite-dimensional Hilbert spaces and linear maps between them. Compact closed categories are known as ``degenerate models'' of linear logic \cite{hyland2003glueing}, though the term degenerate here is not used in a negative way. Therefore a natural question to ask is: ``Is $\mathsf{FHilb}$ a (co)differential category?'' Contrary to what the title of this paper would suggest, the answer to this question is yes, but with a big asterisk. $\mathsf{FHilb}$ admits a trivial (co)differential category structure (Definition \ref{trivdef}) where all structure is zero.  So we must refine our question to: ``Does $\mathsf{FHilb}$ have a non-trivial (co)differential category structure?'' And the answer to this question is no! 

While readers familiar with differential categories may not be surprised by the fact that $\mathsf{FHilb}$ is not an interesting (co)differential category, the reason as to why this is the case is quite interesting. This phenomenon occurs due to an incompatibility between the compact closed structure, specifically the induced trace on $\mathsf{FHilb}$ and any potential codifferential category structure (Theorem \ref{mainthm}). As pointed out to the author by Vicary in private conversations, the proof of this incompatibility is similar to the proof that the Weyl algebra in two variables over a field of characteristic zero does not have any finite dimensional representations \cite{fenn2007weyl}, where the proof makes use of the standard trace of matrices. Again readers familiar with differential categories may point out that this was to be expected since most examples of (co)differential categories \cite{blute2006differential, cockettlemay2018, ehrhard2017introduction} require some sort of ``infinite dimensionality'', in particular for the exponential of linear logic. However, we will see in Example \ref{diffex}.iii that $\mathsf{FVEC}_{\mathbb{Z}_2}$, the category of finite-dimensional vector spaces over $\mathbb{Z}_2$, is a non-trivial codifferential category (which to the knowledge of the author, is a new observation). Therefore, there are examples of non-trivial compact closed/traced (co)differential categories. Digging deeper, one realizes that the incompatibility in $\mathsf{FHilb}$ follows both from $\mathbb{C}$ having characteristic zero and additive inverses. This also explains why $\mathsf{FVEC}_{\mathbb{Z}_2}$ and $\mathsf{REL}$, the category of sets and relations, are both non-trivial codifferential categories (Example \ref{failex}): $\mathbb{Z}_2$ has characteristic $2$, while $\mathsf{REL}$ does not have additive inverses. 

Another goal of this paper is to provide an introduction to differential categories to the categorical quantum foundation community. Hopefully, readers will become interested in differential categories and open the door to further study towards combining these two fields together. 

\textbf{Summary of Paper:} Section 2 provides a detailed overview of codifferential categories by reviewing the definitions and providing detailed examples of additive symmetric monoidal categories (Definition \ref{ASMC} and Example \ref{addex}), algebra modalities (Definition \ref{algmod} and Example \ref{Symex}), and of course codifferential categories (Definition \ref{diffcatdef} and Example \ref{diffex}). We also discuss trivial codifferential category structure (Definition \ref{trivdef}) and explain why $\mathsf{FHilb}$ is a trivial codifferential category. Section 3 is dedicated to proving that the only codifferential category structure on $\mathsf{FHilb}$ is the trivial one. We review traced symmetric monoidal categories (Definition \ref{tracedef}) and consider an additive version of such categories (Definition \ref{addtdef}). Then we provide the main result of this paper (Theorem \ref{mainthm}) and explain why $\mathsf{FHilb}$ is not an interesting (co)differential category (Corollary \ref{maincor}). We also discuss $\mathsf{REL}$ and $\mathsf{FVEC}_{\mathbb{Z}_2}$ which are examples of non-trivial traced codifferential categories and why the main result does not apply to these examples (Example \ref{failex}). We conclude this paper in Section 4 with a discussion of potential directions and future work, which can be summarized as (i) looking for other examples of (co)differential categories with relations to categorial quantum foundations, (ii) studying the other aspects of the theory of differential categories with applications to categorial quantum foundations, and (iii) studying further traced/compact closed (co)differential categories.    

\textbf{Conventions:} Symmetric monoidal categories will be denoted as triples $(\mathbb{X}, \otimes, K)$ where $\mathbb{X}$ is the category, $\otimes$ is the monoidal product, and $K$ is the monoidal unit. We suppress the associativity, unit, and symmetry isomorphisms since we will be working with the graphical calculus \cite{selinger2010survey} of symmetric monoidal categories to express algebraic expressions. String diagrams are to be read left to right. 

\section{Codifferential Categories}

In this section we review \textbf{codifferential categories}, the dual notion of differential categories \cite{blute2006differential}. We have chosen to work with codifferential categories in this paper since the axioms and examples of codifferential categories may be more intuitive for those new to the subject. 

Two of the most important and basic identities from classical differential calculus require addition and zero: the Leibniz rule, $(f(x)g(x))^\prime = f^\prime(x) g(x) + f(x)g^\prime(x)$, and the constant rule, $c^\prime = 0$. Therefore we must discuss the notion of a category having additive structure. Following the terminology used in the differential category literature, here we mean ``additive'' in the sense of Blute, Cockett, and Seely \cite{blute2006differential}, that is, to be enriched over commutative monoids. In particular, this definition does not assume negatives (in the sense of having additive inverses) nor does it assume biproducts, which differs from other sources such as in \cite{mac2013categories}. 

\begin{definition}\label{ASMC} An \textbf{additive category} \cite{blute2006differential} is a category $\mathbb{X}$ such that each hom-set $\mathbb{X}(A,B)$ is a commutative monoid (in the classical sense) with binary operation $+: \mathbb{X}(A,B) \times \mathbb{X}(A,B) \to \mathbb{X}(A,B)$ and unit ${0: A \to B}$, and furthermore that composition preserves the additive structure in the sense that:
\[   \begin{array}[c]{c}\resizebox{!}{1.5cm}{%
\begin{tikzpicture}
	\begin{pgfonlayer}{nodelayer}
		\node [style=function2] (1) at (-9.25, 7) {$f$};
		\node [style=none] (2) at (-10, 7) {};
		\node [style=none] (3) at (-6, 7) {};
		\node [style=function1] (4) at (-8, 7) {$g+h$};
		\node [style=object] (6) at (-5.5, 7) {$=$};
		\node [style=function2] (29) at (-6.75, 7) {$k$};
		\node [style=function2] (30) at (-4.25, 7) {$f$};
		\node [style=none] (31) at (-5, 7) {};
		\node [style=none] (32) at (-1.5, 7) {};
		\node [style=function2] (33) at (-3.25, 7) {$g$};
		\node [style=function2] (34) at (-2.25, 7) {$k$};
		\node [style=object] (35) at (-1, 7) {$+$};
		\node [style=function2] (36) at (0.25, 7) {$f$};
		\node [style=none] (37) at (-0.5, 7) {};
		\node [style=none] (38) at (3, 7) {};
		\node [style=function2] (39) at (1.25, 7) {$h$};
		\node [style=function2] (40) at (2.25, 7) {$k$};
		\node [style=function2] (41) at (-6.75, 6) {$f$};
		\node [style=none] (42) at (-7.5, 6) {};
		\node [style=none] (43) at (-5, 6) {};
		\node [style=object] (45) at (-4.5, 6) {$=$};
		\node [style=function2] (46) at (-5.75, 6) {$0$};
		\node [style=none] (48) at (-4, 6) {};
		\node [style=none] (49) at (-2.25, 6) {};
		\node [style=function2] (50) at (-3.25, 6) {$0$};
		\node [style=object] (52) at (-1.75, 6) {$=$};
		\node [style=function2] (53) at (0.5, 6) {$f$};
		\node [style=none] (54) at (1.25, 6) {};
		\node [style=none] (55) at (-1.25, 6) {};
		\node [style=function2] (56) at (-0.5, 6) {$0$};
		\node [style=none] (57) at (-3, 6) {};
	\end{pgfonlayer}
	\begin{pgfonlayer}{edgelayer}
		\draw [style=wire] (2.center) to (1);
		\draw [style=wire] (1) to (4);
		\draw [style=wire] (4) to (29);
		\draw [style=post] (29) to (3.center);
		\draw [style=wire] (31.center) to (30);
		\draw [style=wire] (30) to (33);
		\draw [style=wire] (33) to (34);
		\draw [style=post] (34) to (32.center);
		\draw [style=wire] (37.center) to (36);
		\draw [style=wire] (36) to (39);
		\draw [style=wire] (39) to (40);
		\draw [style=post] (40) to (38.center);
		\draw [style=wire] (42.center) to (41);
		\draw [style=post] (46) to (43.center);
		\draw [style=wire] (41) to (46);
		\draw [style=pre] (54.center) to (53);
		\draw [style=wire] (56) to (55.center);
		\draw [style=wire] (53) to (56);
		\draw [style=wire] (48.center) to (50);
		\draw [style=post] (50) to (49.center);
	\end{pgfonlayer}
\end{tikzpicture}
  }%
\end{array}
\]
An \textbf{additive symmetric monoidal category} \cite{blute2006differential} is a symmetric monoidal category $(\mathbb{X}, \otimes, K)$ such that $\mathbb{X}$ is an additive category and the monoidal product $\otimes$ is compatible with the additive structure:
\[   \begin{array}[c]{c}\resizebox{!}{4.5cm}{%
\begin{tikzpicture}
	\begin{pgfonlayer}{nodelayer}
		\node [style=object] (6) at (-5.5, 7.5) {$=$};
		\node [style=object] (35) at (-1.5, 7.5) {$+$};
		\node [style=none] (58) at (-9, 6.5) {};
		\node [style=none] (59) at (-6, 6.5) {};
		\node [style=function2] (60) at (-7.5, 6.5) {$k$};
		\node [style=none] (61) at (-8, 6.5) {};
		\node [style=none] (62) at (-9, 7.5) {};
		\node [style=none] (63) at (-6, 7.5) {};
		\node [style=function1] (64) at (-7.5, 7.5) {$g+h$};
		\node [style=none] (65) at (-8, 7.5) {};
		\node [style=none] (66) at (-9, 8.5) {};
		\node [style=none] (67) at (-6, 8.5) {};
		\node [style=function2] (68) at (-7.5, 8.5) {$f$};
		\node [style=none] (69) at (-8, 8.5) {};
		\node [style=none] (70) at (-5, 6.75) {};
		\node [style=none] (71) at (-2, 6.75) {};
		\node [style=function2] (72) at (-3.5, 6.75) {$k$};
		\node [style=none] (73) at (-4, 6.75) {};
		\node [style=none] (74) at (-5, 7.5) {};
		\node [style=none] (75) at (-2, 7.5) {};
		\node [style=function2] (76) at (-3.5, 7.5) {$g$};
		\node [style=none] (77) at (-4, 7.5) {};
		\node [style=none] (78) at (-5, 8.25) {};
		\node [style=none] (79) at (-2, 8.25) {};
		\node [style=function2] (80) at (-3.5, 8.25) {$f$};
		\node [style=none] (81) at (-4, 8.25) {};
		\node [style=none] (82) at (-1, 6.75) {};
		\node [style=none] (83) at (2, 6.75) {};
		\node [style=function2] (84) at (0.5, 6.75) {$k$};
		\node [style=none] (85) at (0, 6.75) {};
		\node [style=none] (86) at (-1, 7.5) {};
		\node [style=none] (87) at (2, 7.5) {};
		\node [style=function2] (88) at (0.5, 7.5) {$h$};
		\node [style=none] (89) at (0, 7.5) {};
		\node [style=none] (90) at (-1, 8.25) {};
		\node [style=none] (91) at (2, 8.25) {};
		\node [style=function2] (92) at (0.5, 8.25) {$f$};
		\node [style=none] (93) at (0, 8.25) {};
		\node [style=object] (94) at (-4, 5) {$=$};
		\node [style=none] (95) at (-7.5, 4.25) {};
		\node [style=none] (96) at (-4.5, 4.25) {};
		\node [style=function2] (97) at (-6, 4.25) {$k$};
		\node [style=none] (98) at (-6.5, 4.25) {};
		\node [style=none] (99) at (-7.5, 5) {};
		\node [style=none] (100) at (-4.5, 5) {};
		\node [style=function2] (101) at (-6, 5) {$0$};
		\node [style=none] (102) at (-6.5, 5) {};
		\node [style=none] (103) at (-7.5, 5.75) {};
		\node [style=none] (104) at (-4.5, 5.75) {};
		\node [style=function2] (105) at (-6, 5.75) {$f$};
		\node [style=none] (106) at (-6.5, 5.75) {};
		\node [style=none] (107) at (-3.5, 4.5) {};
		\node [style=none] (108) at (-0.5, 4.5) {};
		\node [style=none] (111) at (-3.5, 5) {};
		\node [style=none] (112) at (-0.5, 5) {};
		\node [style=none] (115) at (-3.5, 5.5) {};
		\node [style=none] (116) at (-0.5, 5.5) {};
		\node [style=none] (117) at (-2.75, 5.5) {};
		\node [style=none] (118) at (-2.75, 5) {};
		\node [style=none] (119) at (-2.75, 4.5) {};
		\node [style=none] (120) at (-1.25, 5.5) {};
		\node [style=none] (121) at (-1.25, 5) {};
		\node [style=none] (122) at (-1.25, 4.5) {};
		\node [style=function3] (123) at (-2, 5) {$0$};
	\end{pgfonlayer}
	\begin{pgfonlayer}{edgelayer}
		\draw [style=wire] (58.center) to (60);
		\draw [style=post] (60) to (59.center);
		\draw [style=wire] (62.center) to (64);
		\draw [style=post] (64) to (63.center);
		\draw [style=wire] (66.center) to (68);
		\draw [style=post] (68) to (67.center);
		\draw [style=wire] (70.center) to (72);
		\draw [style=post] (72) to (71.center);
		\draw [style=wire] (74.center) to (76);
		\draw [style=post] (76) to (75.center);
		\draw [style=wire] (78.center) to (80);
		\draw [style=post] (80) to (79.center);
		\draw [style=wire] (82.center) to (84);
		\draw [style=post] (84) to (83.center);
		\draw [style=wire] (86.center) to (88);
		\draw [style=post] (88) to (87.center);
		\draw [style=wire] (90.center) to (92);
		\draw [style=post] (92) to (91.center);
		\draw [style=wire] (95.center) to (97);
		\draw [style=post] (97) to (96.center);
		\draw [style=wire] (99.center) to (101);
		\draw [style=post] (101) to (100.center);
		\draw [style=wire] (103.center) to (105);
		\draw [style=post] (105) to (104.center);
		\draw [style=wire] (115.center) to (117.center);
		\draw [style=wire] (111.center) to (118.center);
		\draw [style=wire] (107.center) to (119.center);
		\draw [style=post] (120.center) to (116.center);
		\draw [style=post] (121.center) to (112.center);
		\draw [style=post] (122.center) to (108.center);
	\end{pgfonlayer}
\end{tikzpicture}
  }%
\end{array}
\]
\end{definition}

Note that in the above string diagrams, $+$ is an external operator. For example, the first string diagram equality is the equality $k \circ (g+h) \circ f = (k \circ g \circ f) + (k \circ h \circ f)$. 

\begin{example}\label{addex} \normalfont Here are some examples of well known additive symmetric monoidal categories. We note that in these examples, the additive structures are induced from finite biproducts (see \cite{mac2013categories} for how biproducts induce an additive category structure).
\begin{enumerate}[{\em (i)}]
\item Let $\mathsf{REL}$ be the category of sets and relations. Then $\mathsf{REL}$ is an additive symmetric monoidal category where the monoidal product is given by the Cartesian product of sets and the monoidal unit is a chosen singleton $\lbrace \ast \rbrace$, and where the sum of relations $R,S \subseteq X \times Y$ is defined as their union $R + S := R \cup S \subseteq X \times Y$, and the zero maps are the empty relations $0 := \emptyset \subseteq X \times Y$. 
\item Let $\mathbb{K}$ be a field and let $\mathsf{VEC}_\mathbb{K}$ to be the category of all $\mathbb{K}$-vector spaces and $\mathbb{K}$-linear maps between them. Then $\mathsf{VEC}_\mathbb{K}$ is an additive symmetric monoidal category where the monoidal product is given by the algebraic tensor product of vector spaces and the monoidal unit is $\mathbb{K}$, while the sum of $\mathbb{K}$-linear maps $f, g: V \to W$ is the standard pointwise sum of linear maps, $(f+g)(v) := f(v) + g(v)$, and where the zero maps $0: V \to W$ are the $\mathbb{K}$-linear maps which map everything to zero.
\item Similarly, let $\mathsf{FVEC}_\mathbb{K}$ be the category of all \emph{finite dimensional} $\mathbb{K}$-vector spaces and $\mathbb{K}$-linear maps between them. Then $\mathsf{FVEC}_\mathbb{K}$ is an additive symmetric monoidal category with the same structure defined for $\mathsf{VEC}_\mathbb{K}$.  Perhaps of particular interest to the categorical quantum foundations community is when $\mathbb{K} = \mathbb{C}$, the complex numbers, then famously $\mathsf{FVEC}_\mathbb{C}$ is equivalent to $\mathsf{FHilb}$, the category of finite dimensional Hilbert spaces and $\mathbb{C}$ linear maps between. Therefore $\mathsf{FHilb}$ is also an additive symmetric monoidal category. 
 \end{enumerate}
\end{example} 

The next key ingredient of codifferential category structure is the algebra modality, which is a monad $\mathsf{S}$ whose free $\mathsf{S}$-algebras come equipped with a natural commutative monoid structure. 

\begin{definition}\label{algmod} An \textbf{algebra modality} \cite{blute2006differential, blute2015derivations} on a symmetric monoidal category $(\mathbb{X}, \otimes, K)$ is a quintuple $(\mathsf{S}, \mu, \eta, \mathsf{m}, \mathsf{u})$ consisting of an endofunctor $\mathsf{S}: \mathbb{X} \to \mathbb{X}$ and four natural transformations: 
\[   \begin{array}[c]{c}\resizebox{!}{0.85cm}{%
\begin{tikzpicture}
	\begin{pgfonlayer}{nodelayer}
		\node [style=object] (0) at (3.75, 4) {$\mathsf{S}A$};
		\node [style=object] (1) at (0.75, 4.25) {$\mathsf{S}A$};
		\node [style=object] (2) at (0.75, 3.75) {$\mathsf{S}A$};
		\node [style=function1] (3) at (2.25, 4) {$\mathsf{m}$};
		\node [style=object] (4) at (7, 4) {$\mathsf{S}A$};
		\node [style=function2] (5) at (5.25, 4) {$\mathsf{u}$};
		\node [style=object] (6) at (10.5, 4) {$\mathsf{S}A$};
		\node [style=function2] (7) at (9.25, 4) {$\eta$};
		\node [style=object] (9) at (8.25, 4) {$A$};
		\node [style=object] (10) at (15.5, 4) {$\mathsf{S}A$};
		\node [style=function2] (11) at (14, 4) {$\mu$};
		\node [style=object] (13) at (12.5, 4) {$\mathsf{S}\mathsf{S} A$};
		\node [style=none] (14) at (1.75, 4.25) {};
		\node [style=none] (15) at (1.75, 3.75) {};
	\end{pgfonlayer}
	\begin{pgfonlayer}{edgelayer}
		\draw [style=post, in=180, out=0] (3) to (0);
		\draw [style=post, in=180, out=0] (5) to (4);
		\draw [style=post, in=180, out=0] (7) to (6);
		\draw [style=wire] (9) to (7);
		\draw [style=post, in=180, out=0] (11) to (10);
		\draw [style=wire] (13) to (11);
		\draw [style=wire] (1) to (14.center);
		\draw [style=wire] (2) to (15.center);
	\end{pgfonlayer}
\end{tikzpicture}
  }%
\end{array}
\]
such that:

\noindent \textbf{[AM.1]} $(\mathsf{S}, \mu, \eta)$ is a monad, that is, the following equalities hold: 
\[   \begin{array}[c]{c}\resizebox{!}{1.5cm}{%
\begin{tikzpicture}
	\begin{pgfonlayer}{nodelayer}
		\node [style=function2] (1) at (-8, 7) {$\eta$};
		\node [style=none] (2) at (-9, 7) {};
		\node [style=none] (3) at (-5.5, 7) {};
		\node [style=function2] (4) at (-6.75, 7) {$\mu$};
		\node [style=object] (6) at (-4.75, 7) {$=$};
		\node [style=none] (8) at (-4, 7) {};
		\node [style=none] (9) at (-0.75, 7) {};
		\node [style=object] (15) at (0, 7) {$=$};
		\node [style=function2] (16) at (2, 7) {$\mathsf{S}(\eta)$};
		\node [style=none] (17) at (0.75, 7) {};
		\node [style=none] (18) at (4.5, 7) {};
		\node [style=function2] (19) at (3.5, 7) {$\mu$};
		\node [style=function2] (20) at (-5.25, 5.75) {$\mu$};
		\node [style=none] (21) at (-6.25, 5.75) {};
		\node [style=none] (22) at (-2.75, 5.75) {};
		\node [style=function2] (23) at (-4, 5.75) {$\mu$};
		\node [style=object] (24) at (-2, 5.75) {$=$};
		\node [style=function2] (25) at (0, 5.75) {$\mathsf{S}(\mu)$};
		\node [style=none] (26) at (-1.25, 5.75) {};
		\node [style=none] (27) at (2.75, 5.75) {};
		\node [style=function2] (28) at (1.5, 5.75) {$\mu$};
	\end{pgfonlayer}
	\begin{pgfonlayer}{edgelayer}
		\draw [style=wire] (2.center) to (1);
		\draw [style=post, in=180, out=0] (4) to (3.center);
		\draw [style=wire] (1) to (4);
		\draw [style=post] (8.center) to (9.center);
		\draw [style=wire] (17.center) to (16);
		\draw [style=post, in=180, out=0] (19) to (18.center);
		\draw [style=wire] (16) to (19);
		\draw [style=wire] (21.center) to (20);
		\draw [style=post, in=180, out=0] (23) to (22.center);
		\draw [style=wire] (20) to (23);
		\draw [style=wire] (26.center) to (25);
		\draw [style=post, in=180, out=0] (28) to (27.center);
		\draw [style=wire] (25) to (28);
	\end{pgfonlayer}
\end{tikzpicture}
  }%
\end{array}
\]
\noindent \textbf{[AM.2]} For all objects $A$, $(\mathsf{S}A, \mathsf{m}, \mathsf{u})$ is a commutative monoid in $(\mathbb{X}, \otimes, K)$, that is, the following equalities hold: 
\[   \begin{array}[c]{c}\resizebox{!}{2.75cm}{%
\begin{tikzpicture}
	\begin{pgfonlayer}{nodelayer}
		\node [style=none] (0) at (2.25, 4) {};
		\node [style=none] (1) at (-1.5, 4.25) {};
		\node [style=none] (2) at (-0.5, 3.25) {};
		\node [style=function1] (3) at (1.25, 4) {$\mathsf{m}$};
		\node [style=none] (14) at (0.75, 4.25) {};
		\node [style=none] (15) at (0.75, 3.75) {};
		\node [style=function1] (16) at (-0.5, 3.25) {$\mathsf{m}$};
		\node [style=none] (17) at (-1.5, 3.5) {};
		\node [style=none] (18) at (-1.5, 3) {};
		\node [style=none] (19) at (-1, 3.5) {};
		\node [style=none] (20) at (-1, 3) {};
		\node [style=none] (21) at (0.25, 3.25) {};
		\node [style=none] (22) at (6.75, 3.25) {};
		\node [style=none] (23) at (3, 3) {};
		\node [style=none] (24) at (4, 4) {};
		\node [style=function1] (25) at (5.75, 3.25) {$\mathsf{m}$};
		\node [style=none] (26) at (5.25, 3) {};
		\node [style=none] (27) at (5.25, 3.5) {};
		\node [style=function1] (28) at (4, 4) {$\mathsf{m}$};
		\node [style=none] (29) at (3, 3.75) {};
		\node [style=none] (30) at (3, 4.25) {};
		\node [style=none] (31) at (3.5, 3.75) {};
		\node [style=none] (32) at (3.5, 4.25) {};
		\node [style=none] (33) at (4.75, 4) {};
		\node [style=object] (34) at (2.5, 3.5) {$=$};
		\node [style=none] (35) at (9, 3.5) {};
		\node [style=function1] (36) at (9, 3.5) {$\mathsf{m}$};
		\node [style=none] (37) at (8, 3.25) {};
		\node [style=none] (38) at (8, 3.75) {};
		\node [style=none] (39) at (8.5, 3.25) {};
		\node [style=none] (40) at (8.5, 3.75) {};
		\node [style=none] (41) at (10, 3.5) {};
		\node [style=object] (42) at (10.5, 3.5) {$=$};
		\node [style=none] (43) at (12.5, 3.5) {};
		\node [style=function1] (44) at (12.5, 3.5) {$\mathsf{m}$};
		\node [style=none] (45) at (11.75, 3.25) {};
		\node [style=none] (46) at (11.75, 3.75) {};
		\node [style=none] (47) at (12, 3.25) {};
		\node [style=none] (48) at (12, 3.75) {};
		\node [style=none] (49) at (13.5, 3.5) {};
		\node [style=none] (50) at (11, 3.25) {};
		\node [style=none] (51) at (11, 3.75) {};
		\node [style=none] (52) at (3.75, 1.75) {};
		\node [style=function1] (53) at (3.75, 1.75) {$\mathsf{m}$};
		\node [style=none] (54) at (2.5, 1.5) {};
		\node [style=function2] (55) at (2.5, 2) {$\mathsf{u}$};
		\node [style=none] (56) at (3.25, 1.5) {};
		\node [style=none] (57) at (3.25, 2) {};
		\node [style=none] (58) at (4.75, 1.75) {};
		\node [style=object] (59) at (5.25, 1.75) {$=$};
		\node [style=none] (60) at (5.75, 1.75) {};
		\node [style=none] (61) at (7.5, 1.75) {};
		\node [style=object] (62) at (8, 1.75) {$=$};
		\node [style=none] (63) at (9.75, 1.75) {};
		\node [style=function1] (64) at (9.75, 1.75) {$\mathsf{m}$};
		\node [style=none] (65) at (8.5, 2) {};
		\node [style=function2] (66) at (8.5, 1.5) {$\mathsf{u}$};
		\node [style=none] (67) at (9.25, 2) {};
		\node [style=none] (68) at (9.25, 1.5) {};
		\node [style=none] (69) at (10.75, 1.75) {};
	\end{pgfonlayer}
	\begin{pgfonlayer}{edgelayer}
		\draw [style=post, in=180, out=0] (3) to (0.center);
		\draw [style=wire] (17.center) to (19.center);
		\draw [style=wire] (18.center) to (20.center);
		\draw [style=wire] (1.center) to (14.center);
		\draw [style=wire] (16) to (21.center);
		\draw [style=wire] (21.center) to (15.center);
		\draw [style=post, in=-180, out=0] (25) to (22.center);
		\draw [style=wire] (29.center) to (31.center);
		\draw [style=wire] (30.center) to (32.center);
		\draw [style=wire] (23.center) to (26.center);
		\draw [style=wire] (28) to (33.center);
		\draw [style=wire] (33.center) to (27.center);
		\draw [style=wire] (37.center) to (39.center);
		\draw [style=wire] (38.center) to (40.center);
		\draw [style=post] (36) to (41.center);
		\draw [style=wire] (45.center) to (47.center);
		\draw [style=wire] (46.center) to (48.center);
		\draw [style=post] (44) to (49.center);
		\draw [style=wire] (51.center) to (45.center);
		\draw [style=wire] (50.center) to (46.center);
		\draw [style=wire] (54.center) to (56.center);
		\draw [style=wire] (55) to (57.center);
		\draw [style=post] (53) to (58.center);
		\draw [style=post] (60.center) to (61.center);
		\draw [style=wire] (65.center) to (67.center);
		\draw [style=wire] (66) to (68.center);
		\draw [style=post] (64) to (69.center);
	\end{pgfonlayer}
\end{tikzpicture}
  }%
\end{array}
\]
\noindent \textbf{[AM.3]} $\mu$ is a monoid morphism $(\mathbb{X}, \otimes, K)$, that is, the following equalities hold: 
\[    \begin{array}[c]{c}\resizebox{!}{1.25cm}{%
\begin{tikzpicture}
	\begin{pgfonlayer}{nodelayer}
		\node [style=none] (52) at (3.75, 1.75) {};
		\node [style=function1] (53) at (3.75, 1.75) {$\mathsf{m}$};
		\node [style=function2] (54) at (2.25, 1.25) {$\mu$};
		\node [style=function2] (55) at (2.25, 2.25) {$\mu$};
		\node [style=none] (56) at (3.25, 1.5) {};
		\node [style=none] (57) at (3.25, 2) {};
		\node [style=none] (58) at (4.75, 1.75) {};
		\node [style=object] (59) at (5.25, 1.75) {$=$};
		\node [style=none] (63) at (7, 1.75) {};
		\node [style=function1] (64) at (7, 1.75) {$\mathsf{m}$};
		\node [style=none] (65) at (5.75, 2) {};
		\node [style=none] (66) at (5.75, 1.5) {};
		\node [style=none] (67) at (6.5, 2) {};
		\node [style=none] (68) at (6.5, 1.5) {};
		\node [style=function2] (69) at (8, 1.75) {$\mu$};
		\node [style=none] (70) at (1.75, 2.25) {};
		\node [style=none] (71) at (1.75, 1.25) {};
		\node [style=none] (72) at (8.75, 1.75) {};
		\node [style=function2] (73) at (11, 1.75) {$\mu$};
		\node [style=none] (74) at (11.75, 1.75) {};
		\node [style=function2] (75) at (10.25, 1.75) {$\mathsf{u}$};
		\node [style=object] (76) at (12, 1.75) {$=$};
		\node [style=function2] (77) at (12.5, 1.75) {$\mathsf{u}$};
		\node [style=none] (78) at (13.25, 1.75) {};
		\node [style=none] (79) at (2.75, 2.25) {};
		\node [style=none] (80) at (2.75, 1.25) {};
	\end{pgfonlayer}
	\begin{pgfonlayer}{edgelayer}
		\draw [style=post] (53) to (58.center);
		\draw [style=wire] (65.center) to (67.center);
		\draw [style=wire] (66.center) to (68.center);
		\draw [style=wire] (64) to (69);
		\draw [style=wire] (70.center) to (55);
		\draw [style=wire] (71.center) to (54);
		\draw [style=post] (69) to (72.center);
		\draw [style=post] (73) to (74.center);
		\draw [style=wire, in=180, out=0] (75) to (73);
		\draw [style=post] (77) to (78.center);
		\draw [style=wire] (55) to (79.center);
		\draw [style=wire] (54) to (80.center);
		\draw [style=wire] (80.center) to (56.center);
		\draw [style=wire] (79.center) to (57.center);
	\end{pgfonlayer}
\end{tikzpicture}
  }%
\end{array}
\]
\end{definition}

Note that, as explained in \cite{blute2006differential}, algebra modalities are not required to be lax comonoidal monads, as is usually the case in (dual of) models of linear logic \cite{ehrhard2017introduction}.  

\begin{example}\label{Symex} \normalfont Here are some examples of algebra modalities for the examples defined in Example \ref{addex}. Many other examples of algebra modalities can be found in \cite{blute2006differential, cockettlemay2018, ehrhard2017introduction}. 
\begin{enumerate}[{\em (i)}]
\item Free commutative monoids over sets (in the classical sense) induce an algebra modality on $\mathsf{REL}$. Recall that the free commutative monoid over a set $X$ is the set $\mathsf{M}X$ of all finite multisets (also known as finite bags) of elements of $X$. Explicitly, a finite multiset of elements of $X$ is a function $f: X \to \mathbb{N}$ such that the set $\mathsf{supp}(f) := \lbrace x \in X \vert~ f(x) \neq 0 \rbrace$ is finite, and so $\mathsf{M}X= \lbrace f: X \to \mathbb{N} \vert ~ \vert \mathsf{supp}(f) \vert < \infty \rbrace$. The monoid structure on $\mathsf{M}X$ is defined by point-wise addition, $(f+g)(x)=f(x) + g(x)$, while the unit is the empty multiset $\mathsf{u}: X \to \mathbb{N}$ which maps everything to zero, $\mathsf{u}(x) = 0$. Furthermore, for each $x \in X$ there is the Kronecker delta finite multiset $\delta_x: X \to \mathbb{N}$ which maps $y$ to $1$ if $x=y$ and $0$ otherwise. Define the endofunctor $\oc: \mathsf{REL} \to \mathsf{REL}$ as follows: for a set $X$, let $\oc X := \mathsf{M}X$, while for relation $R \subseteq X \times Y$, define $\oc R \subseteq \oc X \times \oc Y$ as ${\oc R := \lbrace (f, g) \rbrace \vert~ \forall x \in \mathsf{supp}(f) \exists\oc y \in \mathsf{supp}(g).~ (f(x),g(y)) \in R \rbrace}$. Now define the following natural transformations: 
\vspace{-5pt}
\begin{align*}
&\eta= \lbrace (x, \delta_x) \vert~ x \in X \rbrace \subseteq X \times \oc X &&& \mu = \left \lbrace (F, \sum \limits_{f \in \mathsf{supp}(F)} f) \vert~ F \in \mathsf{M}\mathsf{M}X  \right \rbrace \subseteq \oc\oc X \times \oc X \\
&\mathsf{u} = \lbrace (\ast, \mathsf{u}) \rbrace \subseteq \lbrace \ast \rbrace \times \oc X &&& \mathsf{m} = \lbrace \left( (f,g), f+g \right) \vert ~ f,g \in \oc X \rbrace \subseteq (\oc X \times \oc X) \times \oc X
\end{align*}
Note that $\mu$ is well-defined since $\mathsf{supp}(F)$ is finite. Then $(\oc, \mu, \eta, \mathsf{m}, \mathsf{u})$ is an algebra modality on $\mathsf{REL}$, and in fact $\oc X$ is the free commutative monoid over $X$ in $\mathsf{REL}$. Alternatively, a finite multiset $f: X \to \mathbb{N}$ can be seen as a monomial in variables $\mathsf{supp}(f) = \lbrace x_1, \hdots, x_n \rbrace$, specifically $x^{f(x_1)}_1x^{f(x_2)}_2 \hdots x^{f(x_n)}_n$. Therefore, $\mu$ and $\eta$ correspond to composition of monomials, while $\mathsf{m}$ and $\mathsf{u}$ correspond to monomial multiplication. For more details, see \cite[Proposition 2.7]{blute2006differential}. 
\item Symmetric algebras \cite[Section 8, Chapter XVI]{lang2002algebra} induce an algebra modality on categories of vector spaces. From a mathematical physics perspective, symmetric algebras correspond to the bosonic Fock space \cite[Chapter 21]{geroch1985mathematical}. So let $\mathbb{K}$ be a field. For a $\mathbb{K}$-vector space $V$, let $S_n(V)$ be the subspace of $V^{\otimes^n}$ generated by the tensor symmetries $v_1 \otimes ... \otimes v_n - v_{\sigma(1)} \otimes ... \otimes v_{\sigma(n)}$, for all $v_i \in V$ and all $n$-permutations $\sigma$. Define the $n$-th symmetric tensor power of $V$ as $\mathsf{Sym}_n(V) := V^{\otimes^n} / S_n(V)$ and let $v_1 \otimes_s \hdots \otimes_s v_n$ be the equivalence class of $v_1 \otimes \hdots \otimes v_n$ in $\mathsf{Sym}_n(V)$, which we refer as pure symmetric tensors. Now define the endofunctor $\mathsf{Sym}: \mathsf{VEC}_{\mathbb{K}} \to \mathsf{VEC}_{\mathbb{K}}$ which maps $V$ to its symmetric algebra $\mathsf{Sym}(V)$ defined as: 
\vspace{-5pt}
\[\mathsf{Sym}(V)= \bigoplus^{\infty}_{n=0} \mathsf{Sym}_n(V)= \mathbb{K} \oplus V \oplus \mathsf{Sym}_2(V) \oplus ...\]
while for a linear transformation $f: V \to W$, $\mathsf{Sym}(f): \mathsf{Sym}(V) \to \mathsf{Sym}(W)$ is defined on pure tensors as follows $\mathsf{Sym}(f)(v_1 \otimes_s \hdots \otimes_s v_n) = f(v_1) \otimes_s \hdots \otimes_s f(v_n)$ which we extend by linearity. Now define the following natural transformations: 
\vspace{-5pt}
\begin{align*}
\eta: V &\to \mathsf{Sym}(V) & \mathsf{u}: \mathbb{K} &\to \mathsf{Sym}(V)   \\
v &\mapsto v & 1 &\mapsto 1 
\end{align*}
\begin{align*}
 \mu: \mathsf{Sym}^2(V) &\to \mathsf{Sym}(V) \\
 (v_1 \otimes_s \hdots \otimes_s v_n) \otimes_s \hdots \otimes_s (w_1 \otimes_s \hdots \otimes_s w_m) &\mapsto v_1 \otimes_s \hdots \otimes_s v_n \otimes_s \hdots \otimes_s w_1 \otimes_s \hdots \otimes_s w_m \end{align*}
 \begin{align*}
 \mathsf{m}: \mathsf{Sym}(V) \otimes \mathsf{Sym}(V) &\to \mathsf{Sym}(V) \\
(v_1 \otimes_s \hdots \otimes_s v_n) \otimes (w_1 \otimes_s \hdots \otimes_s w_m) &\mapsto v_1 \otimes_s \hdots \otimes_s v_n \otimes_s w_1 \otimes_s \hdots \otimes_s w_m
\end{align*}
which we extend by linearity. Then $(\mathsf{Sym}, \mu, \eta, \mathsf{m}, \mathsf{u})$ is an algebra modality on $\mathsf{VEC}_{\mathbb{K}}$, and as in the previous example, $\mathsf{Sym}(V)$ is the free commutative $\mathbb{K}$-algebra over $V$. In particular, if $X$ is a basis of $V$, then $\mathsf{Sym}(V) \cong \mathbb{K}[X]$ as $\mathbb{K}$-algebras (where $\mathbb{K}[X]$ is the polynomial ring over $X$). Therefore, $\mu$ and $\eta$ correspond to polynomial composition, while $\mathsf{m}$ and $\mathsf{u}$ correspond to polynomial multiplication. For more details on this algebra modality, see \cite[Proposition 2.9]{blute2006differential}. 
\item Unlike symmetric algebras, exterior algebras \cite[Section 8, Chapter XIX ]{lang2002algebra} only induce an algebra modality in a particular case. From a mathematical physics perspective, exterior algebras correspond to the fermionic Fock space \cite[Chapter 21]{geroch1985mathematical}. Let $\mathbb{K}$ be a field. For a $\mathbb{K}$-vector space $V$, let $E_n(V)$ be the subspace of $V^{\otimes^n}$ generated by the alternating tensor symmetries $v_1 \otimes ... \otimes v_n - \mathsf{sign}(\sigma)(v_{\sigma(1)} \otimes ... \otimes v_{\sigma(n)})$ for all $v_i \in V$ and all $n$-permutations $\sigma$, and where $\mathsf{sign}(\sigma)$ is the sign of the permutation. Define the $n$-th exterior power of $V$ as $\mathsf{Ext}_n(V) := V^{\otimes^n} / E_n(V)$ and let $v_1 \wedge \hdots \wedge v_n$ be the equivalence class of $v_1 \otimes \hdots \otimes v_n$ in $\mathsf{Ext}_n(V)$, which we refer to as pure wedge products. Note that if $V$ is finite dimensional, then for all $n > \mathsf{dim}(V)$, $\mathsf{Ext}_n (V) = 0$ since in particular $v \wedge v = 0$. Therefore we can define an endofunctor $\mathsf{Ext}: \mathsf{FVEC}_{\mathbb{K}} \to \mathsf{FVEC}_{\mathbb{K}}$ which maps a finite dimensional vector space $V$ to its exterior algebra $\mathsf{Ext}(V)$ defined as follows: 
\[ \mathsf{Ext}(V) :=  \bigoplus^{\mathsf{dim}(V)}_{n=0} \mathsf{Ext}_n(V)= \mathbb{K} \oplus V \oplus \mathsf{Ext}_2(V) \oplus \hdots \oplus \mathsf{Ext}_{\mathsf{dim}(V)}(V) \]
and for a linear transformation $f: V \to W$, $\mathsf{Ext}(f): \mathsf{Ext}(V) \to \mathsf{Ext}(W)$ is defined on pure wedge products as follows: $\mathsf{Ext}(f)(v_1 \wedge \hdots \wedge v_n) = f(v_1) \wedge \hdots \wedge f(v_n)$, 
which we then extend by linearity. We can also define the following natural transformations: 
\begin{align*}
\eta: V &\to \mathsf{Ext}(V) &&& \mathsf{u}: \mathbb{K} &\to \mathsf{Ext}(V)  \\
v &\mapsto v &&& 1 &\mapsto 1 \end{align*}
\begin{align*}
\mathsf{m}: \mathsf{Ext}(V) \otimes \mathsf{Ext}(V) &\to \mathsf{Ext}(V) \\
(v_1 \wedge \hdots \wedge v_n) \otimes (w_1 \wedge \hdots \wedge w_m) &\mapsto v_1 \wedge \hdots \wedge v_n \wedge w_1 \wedge \hdots \wedge w_m  
\end{align*}
However, there are two problems with $\mathsf{Ext}$ being an algebra modality. The first is that $(\mathsf{Ext}(V), \mathsf{m}, \mathsf{u})$ is not a commutative $\mathbb{K}$-algebra but an \emph{anticommutative} $\mathbb{K}$-algebra since $v \wedge w = -w \wedge v$. The second is that due to this anticommutativity, it is not possible (in general) to construct a well defined $\mu: \mathsf{Ext}^2(V) \to \mathsf{Ext}(V)$ with the desired properties. Both of these problems are solved when $\mathbb{K} = \mathbb{Z}_2$, the field of integers modulo $2$, since in this case $1=-1$ and therefore $v \wedge w = w \wedge v$. So now we can define the following natural transformation:
\begin{align*}
\mu: \mathsf{Ext}^2(V) &\to \mathsf{Ext}(V)   \\
(v_1 \wedge \hdots \wedge v_n) \wedge \hdots \wedge (w_1 \wedge \hdots \wedge w_m) &\mapsto v_1 \wedge \hdots \wedge v_n \wedge \hdots \wedge w_1 \wedge \hdots \wedge w_m 
\end{align*}
which we extend by linearity. Then $(\mathsf{Ext}, \mu, \eta, \mathsf{m}, \mathsf{u})$ is an algebra modality on $\mathsf{FVEC}_{\mathbb{Z}_2}$. For more details on this algebra modality, see \cite[Example 2.4]{hyland2003glueing}. 
\end{enumerate}
\end{example}

It is worth noting that both Example \ref{Symex}.i and Example \ref{Symex}.ii are in fact the same general construction of an algebra modality known as the free exponential modality \cite{mellies2009explicit}.  

\begin{definition}\label{diffcatdef} A \textbf{codifferential category} \cite{blute2006differential, blute2015derivations} is an additive symmetric monoidal category $(\mathbb{X}, \otimes, K)$ with an algebra modality $(\mathsf{S}, \mu, \eta, \mathsf{m}, \mathsf{u})$ which comes equipped with a \textbf{deriving transformation}, that is, a natural transformation
\[   \begin{array}[c]{c}\resizebox{!}{0.85cm}{%
\begin{tikzpicture}
	\begin{pgfonlayer}{nodelayer}
		\node [style=none] (63) at (2.75, 1.75) {};
		\node [style=function1] (64) at (2.75, 1.75) {$\mathsf{d}$};
		\node [style=object] (65) at (4.25, 2) {$\mathsf{S}A$};
		\node [style=object] (66) at (4.25, 1.5) {$A$};
		\node [style=none] (67) at (3.25, 2) {};
		\node [style=none] (68) at (3.25, 1.5) {};
		\node [style=object] (72) at (1.25, 1.75) {$\mathsf{S}A$};
	\end{pgfonlayer}
	\begin{pgfonlayer}{edgelayer}
		\draw [style=pre] (65) to (67.center);
		\draw [style=pre] (66) to (68.center);
		\draw [style=wire] (64) to (72);
	\end{pgfonlayer}
\end{tikzpicture}
  }%
\end{array}
\]
such that the following equalities\footnote{It should be noted that the interchange rule was not part of the definition in \cite{blute2006differential} but was later added to ensure that the coKleisli category of a differential category was a Cartesian differential category \cite{blute2009cartesian}.}  hold: 
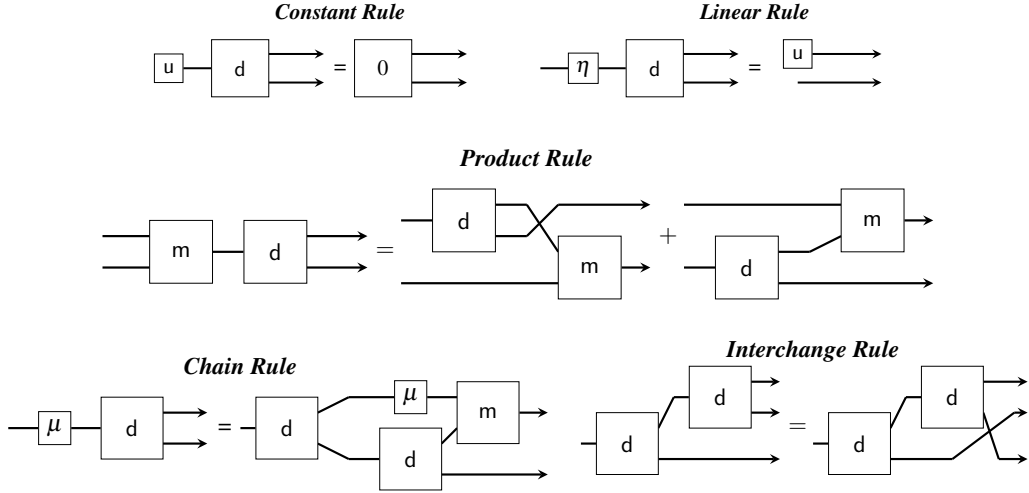
\[   \begin{array}[c]{c}\resizebox{!}{1.25cm}{%
\begin{tikzpicture}
	\begin{pgfonlayer}{nodelayer}
		\node [style=none] (9) at (-4.75, 2) {\textbf{Constant Rule}};
		\node [style=none] (16) at (2.5, 1) {=};
		\node [style=none] (17) at (4.75, 1.25) {};
		\node [style=none] (18) at (4.75, 0.75) {};
		\node [style=function2] (19) at (3.25, 1.25) {$\mathsf{u}$};
		\node [style=none] (20) at (3.25, 0.75) {};
		\node [style=none] (21) at (2.5, 2) {\textbf{Linear Rule}};
		\node [style=none] (22) at (-6.5, 1) {};
		\node [style=function1] (23) at (-6.5, 1) {$\mathsf{d}$};
		\node [style=none] (24) at (-5, 1.25) {};
		\node [style=none] (25) at (-5, 0.75) {};
		\node [style=none] (26) at (-6, 1.25) {};
		\node [style=none] (27) at (-6, 0.75) {};
		\node [style=function2] (28) at (-7.75, 1) {$\mathsf{u}$};
		\node [style=none] (29) at (-4.75, 1) {=};
		\node [style=none] (30) at (-4, 1) {};
		\node [style=function1] (31) at (-4, 1) {$0$};
		\node [style=none] (32) at (-2.5, 1.25) {};
		\node [style=none] (33) at (-2.5, 0.75) {};
		\node [style=none] (34) at (-3.5, 1.25) {};
		\node [style=none] (35) at (-3.5, 0.75) {};
		\node [style=none] (36) at (0.75, 1) {};
		\node [style=function1] (37) at (0.75, 1) {$\mathsf{d}$};
		\node [style=none] (38) at (2.25, 1.25) {};
		\node [style=none] (39) at (2.25, 0.75) {};
		\node [style=none] (40) at (1.25, 1.25) {};
		\node [style=none] (41) at (1.25, 0.75) {};
		\node [style=function2] (42) at (-0.5, 1) {$\eta$};
		\node [style=none] (43) at (-1.25, 1) {};
	\end{pgfonlayer}
	\begin{pgfonlayer}{edgelayer}
		\draw [style=post] (19) to (17.center);
		\draw [style=post] (20.center) to (18.center);
		\draw [style=pre] (24.center) to (26.center);
		\draw [style=pre] (25.center) to (27.center);
		\draw [style=wire] (23) to (28);
		\draw [style=pre] (32.center) to (34.center);
		\draw [style=pre] (33.center) to (35.center);
		\draw [style=pre] (38.center) to (40.center);
		\draw [style=pre] (39.center) to (41.center);
		\draw [style=wire] (37) to (42);
		\draw [style=wire] (43.center) to (42);
	\end{pgfonlayer}
\end{tikzpicture}
  }%
\end{array}
\]
\[   \begin{array}[c]{c}\resizebox{!}{2cm}{%
\begin{tikzpicture}
	\begin{pgfonlayer}{nodelayer}
		\node [style=none] (1) at (12.25, -1.25) {\textbf{Product Rule}};
		\node [style=none] (15) at (14.5, -2.5) {$+$};
		\node [style=none] (22) at (8.25, -2.75) {};
		\node [style=function1] (23) at (8.25, -2.75) {$\mathsf{d}$};
		\node [style=none] (24) at (9.75, -2.5) {};
		\node [style=none] (25) at (9.75, -3) {};
		\node [style=none] (26) at (8.75, -2.5) {};
		\node [style=none] (27) at (8.75, -3) {};
		\node [style=none] (29) at (6.75, -2.75) {};
		\node [style=function1] (30) at (6.75, -2.75) {$\mathsf{m}$};
		\node [style=none] (31) at (5.5, -2.5) {};
		\node [style=none] (32) at (5.5, -3) {};
		\node [style=none] (33) at (6.25, -2.5) {};
		\node [style=none] (34) at (6.25, -3) {};
		\node [style=none] (35) at (11.25, -2.25) {};
		\node [style=function1] (36) at (11.25, -2.25) {$\mathsf{d}$};
		\node [style=none] (37) at (11.75, -2) {};
		\node [style=none] (38) at (11.75, -2.5) {};
		\node [style=none] (39) at (10.25, -2.25) {};
		\node [style=none] (40) at (12.25, -2.5) {};
		\node [style=none] (41) at (13.25, -3) {};
		\node [style=function1] (42) at (13.25, -3) {$\mathsf{m}$};
		\node [style=none] (44) at (10.25, -3.25) {};
		\node [style=none] (45) at (12.75, -2.75) {};
		\node [style=none] (46) at (12.75, -3.25) {};
		\node [style=none] (47) at (14.25, -3) {};
		\node [style=none] (48) at (14.25, -2) {};
		\node [style=none] (49) at (10, -2.75) {$=$};
		\node [style=none] (50) at (12.25, -2) {};
		\node [style=none] (51) at (12.75, -2) {};
		\node [style=none] (52) at (15.75, -3) {};
		\node [style=function1] (53) at (15.75, -3) {$\mathsf{d}$};
		\node [style=none] (54) at (16.25, -2.75) {};
		\node [style=none] (55) at (16.25, -3.25) {};
		\node [style=none] (56) at (14.75, -2) {};
		\node [style=none] (58) at (17.75, -2.25) {};
		\node [style=function1] (59) at (17.75, -2.25) {$\mathsf{m}$};
		\node [style=none] (60) at (14.75, -3) {};
		\node [style=none] (61) at (17.25, -2) {};
		\node [style=none] (62) at (17.25, -2.5) {};
		\node [style=none] (63) at (18.75, -2.25) {};
		\node [style=none] (66) at (16.75, -2.75) {};
		\node [style=none] (67) at (18.75, -3.25) {};
	\end{pgfonlayer}
	\begin{pgfonlayer}{edgelayer}
		\draw [style=pre] (24.center) to (26.center);
		\draw [style=pre] (25.center) to (27.center);
		\draw [style=wire] (31.center) to (33.center);
		\draw [style=wire] (32.center) to (34.center);
		\draw [style=wire] (30) to (23);
		\draw [style=wire] (36) to (39.center);
		\draw [style=wire] (38.center) to (40.center);
		\draw [style=wire] (44.center) to (46.center);
		\draw [style=post] (42) to (47.center);
		\draw [style=wire] (37.center) to (50.center);
		\draw [style=wire] (50.center) to (45.center);
		\draw [style=wire] (40.center) to (51.center);
		\draw [style=post] (51.center) to (48.center);
		\draw [style=post] (59) to (63.center);
		\draw [style=wire] (54.center) to (66.center);
		\draw [style=wire] (60.center) to (53);
		\draw [style=wire] (56.center) to (61.center);
		\draw [style=wire] (66.center) to (62.center);
		\draw [style=post] (55.center) to (67.center);
	\end{pgfonlayer}
\end{tikzpicture}
  }%
\end{array}
\]
\[   \begin{array}[c]{c}\resizebox{!}{2cm}{%
\begin{tikzpicture}
	\begin{pgfonlayer}{nodelayer}
		\node [style=none] (93) at (1.25, -1) {\textbf{Interchange Rule}};
		\node [style=none] (100) at (-8.25, -2.25) {=};
		\node [style=none] (105) at (-8, -1.25) {\textbf{Chain Rule}};
		\node [style=none] (106) at (-9.75, -2.25) {};
		\node [style=function1] (107) at (-9.75, -2.25) {$\mathsf{d}$};
		\node [style=none] (108) at (-8.5, -2) {};
		\node [style=none] (109) at (-8.5, -2.5) {};
		\node [style=none] (110) at (-9.25, -2) {};
		\node [style=none] (111) at (-9.25, -2.5) {};
		\node [style=function2] (112) at (-11, -2.25) {$\mu$};
		\node [style=none] (113) at (-11.75, -2.25) {};
		\node [style=none] (114) at (-7.25, -2.25) {};
		\node [style=function1] (115) at (-7.25, -2.25) {$\mathsf{d}$};
		\node [style=none] (118) at (-6.75, -2) {};
		\node [style=none] (119) at (-6.75, -2.5) {};
		\node [style=none] (120) at (-8, -2.25) {};
		\node [style=none] (121) at (-6.25, -2.75) {};
		\node [style=none] (122) at (-6.25, -1.75) {};
		\node [style=function2] (123) at (-5.25, -1.75) {$\mu$};
		\node [style=none] (124) at (-5.25, -2.75) {};
		\node [style=function1] (125) at (-5.25, -2.75) {$\mathsf{d}$};
		\node [style=none] (127) at (-3, -3) {};
		\node [style=none] (128) at (-4.75, -2.5) {};
		\node [style=none] (129) at (-4.75, -3) {};
		\node [style=none] (130) at (-3, -2) {};
		\node [style=none] (131) at (-4, -2) {};
		\node [style=function1] (132) at (-4, -2) {$\mathsf{m}$};
		\node [style=none] (133) at (-4.5, -1.75) {};
		\node [style=none] (134) at (-4.5, -2.25) {};
		\node [style=none] (147) at (-2.5, -2.5) {};
		\node [style=none] (148) at (0.75, -2.75) {};
		\node [style=none] (149) at (-0.25, -1.75) {};
		\node [style=function1] (150) at (-1.75, -2.5) {$\mathsf{d}$};
		\node [style=none] (151) at (-1.25, -2.75) {};
		\node [style=none] (152) at (-1.25, -2.25) {};
		\node [style=function1] (153) at (-0.25, -1.75) {$\mathsf{d}$};
		\node [style=none] (154) at (0.75, -2) {};
		\node [style=none] (155) at (0.75, -1.5) {};
		\node [style=none] (156) at (0.25, -2) {};
		\node [style=none] (157) at (0.25, -1.5) {};
		\node [style=none] (158) at (-1, -1.75) {};
		\node [style=object] (159) at (1, -2.25) {$=$};
		\node [style=none] (160) at (1.25, -2.5) {};
		\node [style=none] (161) at (4.75, -2) {};
		\node [style=none] (162) at (3.75, -1.75) {};
		\node [style=function1] (163) at (2, -2.5) {$\mathsf{d}$};
		\node [style=none] (164) at (2.5, -2.75) {};
		\node [style=none] (165) at (2.5, -2.25) {};
		\node [style=function1] (166) at (3.5, -1.75) {$\mathsf{d}$};
		\node [style=none] (167) at (4.75, -2.75) {};
		\node [style=none] (168) at (4.75, -1.5) {};
		\node [style=none] (169) at (4, -2) {};
		\node [style=none] (170) at (4, -1.5) {};
		\node [style=none] (171) at (2.75, -1.75) {};
		\node [style=none] (172) at (3.5, -2.75) {};
		\node [style=none] (173) at (4.25, -2.75) {};
		\node [style=none] (174) at (4.5, -2) {};
	\end{pgfonlayer}
	\begin{pgfonlayer}{edgelayer}
		\draw [style=pre] (108.center) to (110.center);
		\draw [style=pre] (109.center) to (111.center);
		\draw [style=wire] (107) to (112);
		\draw [style=wire] (113.center) to (112);
		\draw [style=wire] (115) to (120.center);
		\draw [style=wire] (118.center) to (122.center);
		\draw [style=wire] (119.center) to (121.center);
		\draw [style=wire] (122.center) to (123);
		\draw [style=pre] (127.center) to (129.center);
		\draw [style=wire] (121.center) to (125);
		\draw [style=pre] (130.center) to (132);
		\draw [style=wire] (123) to (133.center);
		\draw [style=wire] (128.center) to (134.center);
		\draw [style=wire, in=0, out=180] (150) to (147.center);
		\draw [style=pre] (154.center) to (156.center);
		\draw [style=pre] (155.center) to (157.center);
		\draw [style=pre] (148.center) to (151.center);
		\draw [style=wire] (153) to (158.center);
		\draw [style=wire] (158.center) to (152.center);
		\draw [style=wire, in=0, out=180] (163) to (160.center);
		\draw [style=pre] (168.center) to (170.center);
		\draw [style=wire] (166) to (171.center);
		\draw [style=wire] (171.center) to (165.center);
		\draw [style=wire] (164.center) to (172.center);
		\draw [style=wire] (169.center) to (173.center);
		\draw [style=post] (173.center) to (167.center);
		\draw [style=wire] (172.center) to (174.center);
		\draw [style=post] (174.center) to (161.center);
	\end{pgfonlayer}
\end{tikzpicture}
  }%
\end{array}
\]
\end{definition}

 The deriving transformation axioms are probably best understood by studying Example \ref{diffex}.ii below, which arises from polynomial differentiation. Briefly, the deriving transformation can be viewed as a sort of internal differential operator which maps $f(x)$ to its derivative $f^\prime(x)\mathsf{d}(x)$. In Example \ref{Symex} it was discussed that $\mu$ and $\eta$ correspond to composition, while $\mathsf{m}$ and $\mathsf{u}$ correspond to multiplication. Therefore the deriving transformation axioms are precisely their namesakes from classical differential calculus. From a Fock space perspective, as explained by Fiore in \cite{fiore2015axiomatics}, the deriving transformation should be interpreted as a creation operator \cite{geroch1985mathematical}. For a more in-depth discussion on the interpretation of these axioms and (co)differential categories in general, we refer to reader to the original paper on differential categories \cite{blute2006differential}. 

\begin{example}\label{diffex} \normalfont Here are some examples of codifferential categories with the algebra modalities from Example \ref{Symex}. Many other examples of (co)differential categories can also be found in \cite{blute2006differential, cockettlemay2018, ehrhard2017introduction}.
\begin{enumerate}[{\em (i)}]
\item $\mathsf{REL}$ is a codifferential category with algebra modality $(\oc, \mu, \eta, \mathsf{m}, \mathsf{u})$ (as defined in Example \ref{Symex}.i) and with deriving transformation $\mathsf{d} \subseteq \oc X \times (\oc X \times X)$ defined as follows:
\[\mathsf{d} := \lbrace \left( f, (f - \delta_x, x)\right) \vert~ x \in \mathsf{supp}(f) \rbrace \subset \oc X \times (\oc X \times X) \]
Note that for $x \in \mathsf{supp}(f)$, the finite multiset $f - \delta_x$ is well defined since if $y \neq x$, $f(y) - \delta_x(y) = f(y)$, while for $x$, $f(x) \neq 0$ which implies that $f(x) = n+1$ for some $n$, and therefore $f(x) - \delta_x(x) = n$. Viewing finite multisets as monomials, the deriving transformation is the relation which relates $x^{k_1}_1 x^{k_2}_2 \hdots x^{k_n}_n$ to $(x^{k_1}_1 x^{k_2}_2 \hdots x^{k_i-1}_i \hdots x^{k_n}_n, x_i)$.  See \cite[Proposition 2.7]{blute2006differential} for more details on this example. 
\item Let $\mathbb{K}$ be a field. Then $\mathsf{VEC}_\mathbb{K}$ is a codifferential category with the algebra modality $(\mathsf{Sym}, \mu, \eta, \mathsf{m}, \mathsf{u})$ (as defined in Example \ref{Symex}.ii) and with deriving transformation defined on pure tensors as follows: 
\begin{align*}
\mathsf{d}: \mathsf{Sym}(V) &\to \mathsf{Sym}(V) \otimes V  \\
v_1 \otimes_s \hdots \otimes_s v_n &\mapsto \sum \limits^n_{i=1} (v_1 \otimes_s \hdots \otimes_s v_{i-1} \otimes_s v_{i+1} \otimes_s \hdots \otimes_s v_n) \otimes v_i 
\end{align*}
which we then extend by linearity. In particular if $X= \lbrace x_1, \hdots \rbrace$ is a basis for $V$, then the deriving transformation can alternatively be described as a map $\mathsf{d}: \mathbb{K}[X] \to \mathbb{K}[X] \otimes V$ which is given by taking the sum of partial derivatives: 
\vspace{-5pt}
\[\mathsf{d}(\mathsf{p}(x_1, \hdots, x_n))=  \sum_{i=1}^{n} \frac{\partial \mathsf{p}}{\partial x_i}(x_1, \hdots, x_n) \otimes x_i \]
See \cite[Proposition 2.9]{blute2006differential} for more details on this example. 
\item $\mathsf{FVEC}_{\mathbb{Z}_2}$ is a codifferential category with the algebra modality $(\mathsf{Ext}, \mu, \eta, \mathsf{m}, \mathsf{u})$ (as defined in Example \ref{Symex}.iii) and with deriving transformation defined on pure wedge products as follows: 
\begin{align*}
\mathsf{d}: \mathsf{Ext}(V) &\to \mathsf{Ext}(V) \otimes V  \\
v_1 \wedge \hdots \wedge v_n &\mapsto \sum \limits^n_{i=1} (v_1 \wedge \hdots \wedge v_{i-1} \wedge v_{i+1} \wedge \hdots \wedge v_n) \otimes v_i 
\end{align*}
which we then extend by linearity. Furthermore, note that the $\mathsf{d}$ could have been defined in $\mathsf{FVEC}_\mathbb{K}$ for arbitrary $\mathbb{K}$, with the necessary scalar multiplication by permutation signs. 
\end{enumerate}
\end{example}

As mentioned in the introduction, recall that the title of the paper is ``$\mathsf{FHilb}$ is not an \emph{interesting} (co)differential category'' rather than ``$\mathsf{FHilb}$ is not a (co)differential category''. This is because there is a trivial (and therefore uninteresting) codifferential category structure on $\mathsf{FHilb}$ induced by its zero object. Recall that a zero object in a category $\mathbb{X}$ is an object $Z$ such that for every object $A$ there is a unique map to and from $Z$. Note that if $\mathbb{X}$ is an additive category, these unique maps must be $0$. 

\begin{definition}\label{trivdef} \normalfont A codifferential category with algebra modality  $(\mathsf{S}, \mu, \eta, \mathsf{m}, \mathsf{u})$ and deriving transformation $\mathsf{d}$ is said to be \textbf{trivial} if for each object $A$, $\mathsf{S}A$ is a zero object. 
\end{definition}

\begin{example} \normalfont Any additive symmetric monoidal category with a zero object $Z$ is a trivial codifferential category setting $\mathsf{S}(-) = Z$, and $\mu=0$, $\eta=0$, $\mathsf{m}=0$, $\mathsf{u}=0$, and $\mathsf{d}=0$. In particular, $\mathsf{FVEC}_{\mathbb{K}}$ (for any field $\mathbb{K}$) and $\mathsf{FHilb}$ both have zero objects, the zero vector space, and therefore are trivial codifferential categories. The goal is now to show that this is the only codifferential category structure on $\mathsf{FHilb}$! 
\end{example} 

\section{$\mathsf{FHilb}$ is not an interesting codifferential category}

As explained in the introduction, there is an incompatibility between possible deriving transformations on $\mathsf{FHilb}$ and the trace. And as we alluded to earlier, the proof of this result is similar to the proof that the Weyl algebra in two variables over a field of characteristic zero does not have any finite dimensional representations:

\begin{proposition}\cite{fenn2007weyl}\label{ST} Let $\mathbb{K}$ be a field of characteristic zero and let $\mathsf{W} := \mathbb{K}(x,y)/\langle xy-yx -1 \rangle$, where $\mathbb{K}(x,y)$ is the non-commutative polynomial ring. Then there exists a $\mathbb{K}$-algebra morphism \\ \noindent ${\sigma: \mathsf{W} \to \mathsf{MAT}(\mathbb{K})_n}$ if and only if $n=0$. 
\end{proposition} 
\vspace{-5pt}
\begin{proof} Let $\sigma: W \to \mathsf{MAT}(\mathbb{K})_n$ be a $\mathbb{K}$-algebra morphism. Then we have that:
\vspace{-5pt}
\begin{align*}
n &=~ \mathsf{Tr}(\mathsf{I}_n) =~ \mathsf{Tr}(\sigma(1)) =~ \mathsf{Tr}(\sigma(xy-yx)) =~ \mathsf{Tr}(\sigma(x)\sigma(y) - \sigma(y)\sigma(x))\\
&=~ \mathsf{Tr}(\sigma(x)\sigma(y)) - \mathsf{Tr}(\sigma(y)\sigma(x)) =~ \mathsf{Tr}(\sigma(x)\sigma(y)) - \mathsf{Tr}(\sigma(x)\sigma(y)) =~ 0 
\end{align*}
\end{proof} 

The above proof required (i) a trace, (ii) additive inverses, (iii) that $xy-yx=1$, and (iv) that the trace of the identity matrix is equal to its dimension. In particular regarding this last ingredient, notice that Proposition \ref{ST} fails for fields of non-zero characteristic -- but more on this later (Example \ref{failex}.ii). We begin by first briefly reviewing the notion of traced symmetric monoidal categories and consider their additive version. 

Traced monoidal categories, introduced by Joyal, Street, and Verity \cite{joyal1996traced}, are balanced monoidal categories equipped with a generalization of the classical notion of partial traces for finite-dimensional vector spaces. In this paper, we only need to concern ourselves with the symmetric case, and so we provide the (slightly modified) axiomatization and graphical representation of traced symmetric monoidal categories as found in \cite{hasegawa1997recursion}. Alternative (but equivalent) axiomatizations can be found in \cite{hasegawa2009traced,hasegawa2008finite}. 

\begin{definition}\label{tracedef} A \textbf{traced symmetric monoidal category} \cite{hasegawa1997recursion} is a symmetric monoidal category $\mathbb{X}$ \\ \noindent equipped with a family of functions (for each triple of objects $X, A, B$ of $\mathbb{X}$):
\begin{equation}\label{}\begin{gathered} \mathsf{Tr}^X: \mathbb{X}(X \otimes A, X \otimes B) \to \mathbb{X}(A, B)
 \end{gathered}\end{equation}
 where for a map $f: X \otimes A \to X \otimes B$, its image $\mathsf{Tr}^X(f): A \to B$ is represented graphically as:
 \[    \begin{array}[c]{c}\resizebox{!}{1.15cm}
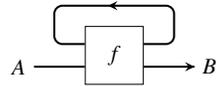
{%
\begin{tikzpicture}
\node[myblock]
  (G) 
  {$f$};
\draw[rounded corners,onearrow={0.525}{}]
  ([shift={(0pt,20pt)}]G.south east) -|
  ([shift={(15pt,10pt)}]G.north east) --
  ([shift={(-15pt,10pt)}]G.north west) |-
  ([shift={(0pt,20pt)}]G.south west);
\draw[zeroarrow={0}{$A$}]
  ([shift={(-25pt,-20pt)}]G.north west) -- ([shift={(0pt,-20pt)}]G.north west);
\draw[onearrow={1}{$B$}]
  ([shift={(0pt,-20pt)}]G.north east) -- ([shift={(25pt,-20pt)}]G.north east);
\end{tikzpicture}
  }%
\end{array} \]
and such that the following equalities are satisfied: 

\noindent ~~~~~~~~~~~~~~~~~~~~~~~~~~~~~~~~~~~~~~~~~~~~~~~~~~~~~~~~~~~~~~~~~~~~~~~~~~~\textbf{Tightening}:
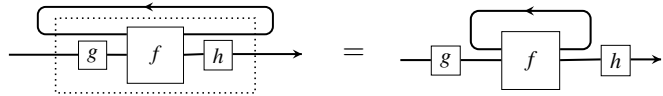
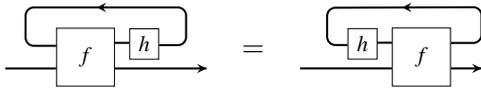
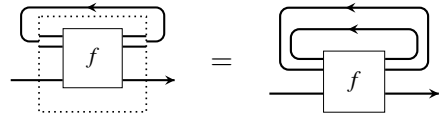
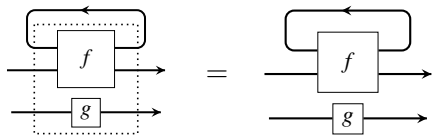
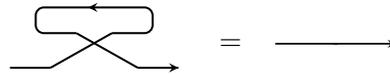
\[\begin{array}[c]{c}\resizebox{!}{1.25cm}{%
\begin{tikzpicture}
\node[myblock]
  (G) 
  {$f$};
    \node[mysquare,left=10pt of {G}] (S) {$g$}; 
      \node[mysquare,right=10pt of {G}] (T) {$h$}; 
\draw[rounded corners,onearrow={0.525}{}]
  ([shift={(0pt,25pt)}]G.south east) -|
  ([shift={(45pt,10pt)}]G.north east) --
  ([shift={(-45pt,10pt)}]G.north west) |-
  ([shift={(0pt,25pt)}]G.south west);
\draw[zeroarrow={0}{}]
  ([shift={(-35pt,-0.25cm)}]S.north west) -- ([shift={(0pt,-0.25cm)}]S.north west);
\draw[zeroarrow={0}{}]
  ([shift={(0pt,-0.25cm)}]S.north east) -- ([shift={(0pt,-0.5cm)}]G.north west);
\draw[zeroarrow={0}{}]
  ([shift={(0pt,-0.5cm)}]G.north east) -- ([shift={(0pt,-0.25cm)}]T.north west);
  \draw[onearrow={1}{}]
  ([shift={(0pt,-0.25cm)}]T.north east) -- ([shift={(35pt,-0.25cm)}]T.north east);
  \draw[thick,dotted] ($(S.north west)+(-0.4,0.4)$)  rectangle ($(T.south east)+(0.4,-0.4)$);
\end{tikzpicture} }%
   \end{array}=
   \begin{array}[c]{c}\resizebox{!}{1.15cm}{%
\begin{tikzpicture}
\node[myblock]
  (G) 
  {$f$};
    \node[mysquare,left=20pt of {G}] (S) {$g$}; 
      \node[mysquare,right=20pt of {G}] (T) {$h$}; 
\draw[rounded corners,onearrow={0.525}{}]
  ([shift={(0pt,20pt)}]G.south east) -|
  ([shift={(15pt,10pt)}]G.north east) --
  ([shift={(-15pt,10pt)}]G.north west) |-
  ([shift={(0pt,20pt)}]G.south west);
\draw[zeroarrow={0}{}]
  ([shift={(-15pt,-0.25cm)}]S.north west) -- ([shift={(0pt,-0.25cm)}]S.north west);
\draw[zeroarrow={0}{}]
  ([shift={(0pt,-0.25cm)}]S.north east) -- ([shift={(0pt,-0.5cm)}]G.north west);
\draw[zeroarrow={0}{}]
  ([shift={(0pt,-0.5cm)}]G.north east) -- ([shift={(0pt,-0.25cm)}]T.north west);
  \draw[onearrow={1}{}]
  ([shift={(0pt,-0.25cm)}]T.north east) -- ([shift={(15pt,-0.25cm)}]T.north east);
\end{tikzpicture} 
  }%
   \end{array}\]
~~~~~~~~~~~~~~~~~~~~~~~~~~~~~~~~~~~~~~~~~ \textbf{Sliding} ~~~~~~~~~~~~~~~~~~~~~~~~~~~~~~~~~~~~~~~~~~~~~~~~~~~~~~~~~~~~~ \textbf{Vanishing}
     \[ \begin{array}[c]{c}\resizebox{!}{1.15cm}{%
 \begin{tikzpicture}
\node[myblock](G) {$f$};
      \node[mysquare,above right=-0.55cm and 0.25cm of {G}] (T) {$h$}; 
\draw[rounded corners,onearrow={0.525}{}]
  ([shift={(0cm,0.25cm)}]T.south east) -|
  ([shift={(35pt,10pt)}]G.north east) --
  ([shift={(-15pt,10pt)}]G.north west) |-
  ([shift={(0pt,20pt)}]G.south west);
\draw[zeroarrow={0}{}]
  ([shift={(-25pt,-20pt)}]G.north west) -- ([shift={(0pt,-20pt)}]G.north west);
  \draw[zeroarrow={0}{}]
  ([shift={(0pt,-0.25cm)}]T.north west) -- ([shift={(0pt,-0.26cm)}]G.north east);
\draw[onearrow={1}{}]
  ([shift={(0pt,-20pt)}]G.north east) -- ([shift={(45pt,-20pt)}]G.north east);
\end{tikzpicture}}%
   \end{array}=
   \begin{array}[c]{c}\resizebox{!}{1.15cm}{%
 \begin{tikzpicture}
\node[myblock](G) {$f$};
      \node[mysquare,above left=-0.55cm and 0.25cm of {G}] (T) {$h$}; 
\draw[rounded corners,onearrow={0.525}{}]
  ([shift={(0cm,0.75cm)}]G.south east) -|
  ([shift={(15pt,10pt)}]G.north east) --
  ([shift={(-10pt,10pt)}]T.north west) |-
  ([shift={(0pt,0.25cm)}]T.south west);
\draw[zeroarrow={0}{}]
  ([shift={(-45pt,-20pt)}]G.north west) -- ([shift={(0pt,-20pt)}]G.north west);
  \draw[zeroarrow={0}{}]
  ([shift={(0pt,-0.25cm)}]T.north east) -- ([shift={(0pt,-0.26cm)}]G.north west);
\draw[onearrow={1}{}]
  ([shift={(0pt,-20pt)}]G.north east) -- ([shift={(15pt,-20pt)}]G.north east);
\end{tikzpicture}}%
   \end{array} ~~~~~~~~~~~  \begin{array}[c]{c}\resizebox{!}{1.5cm}{%
\begin{tikzpicture}
\node[myblock]
  (G) 
  {$f$};
\draw[rounded corners,onearrow={0.525}{}]
  ([shift={(0.4,22.5pt)}]G.south east) -|
  ([shift={(20pt,10pt)}]G.north east) --
  ([shift={(-20pt,10pt)}]G.north west) |-
  ([shift={(-0.4,22.5pt)}]G.south west);
\draw[zeroarrow={0}{}]
  ([shift={(-25pt,-25pt)}]G.north west) -- ([shift={(0pt,-25pt)}]G.north west);
\draw[onearrow={1}{}]
  ([shift={(0pt,-25pt)}]G.north east) -- ([shift={(25pt,-25pt)}]G.north east);
  \draw[zeroarrow={0}{}]
  ([shift={(-0.4,25pt)}]G.south west) -- ([shift={(0pt,25pt)}]G.south west);
\draw[zeroarrow={1}{}]
  ([shift={(0pt,25pt)}]G.south east) -- ([shift={(0.4,25pt)}]G.south east);
    \draw[zeroarrow={0}{}]
  ([shift={(-0.4,20pt)}]G.south west) -- ([shift={(0pt,20pt)}]G.south west);
\draw[zeroarrow={1}{}]
  ([shift={(0pt,20pt)}]G.south east) -- ([shift={(0.4,20pt)}]G.south east);
    \draw[thick,dotted] ($(G.north west)+(-0.4,0.2)$)  rectangle ($(G.south east)+(0.4,-0.4)$);
\end{tikzpicture}}%
   \end{array}=
   \begin{array}[c]{c}\resizebox{!}{1.5cm}{%
\begin{tikzpicture}
\node[myblock]
  (G) 
  {$f$};
\draw[rounded corners,onearrow={0.525}{}]
  ([shift={(0pt,20pt)}]G.south east) -|
  ([shift={(20pt,20pt)}]G.north east) --
  ([shift={(-20pt,20pt)}]G.north west) |-
  ([shift={(0pt,20pt)}]G.south west);
  \draw[rounded corners,onearrow={0.525}{}]
  ([shift={(0pt,25pt)}]G.south east) -|
  ([shift={(15pt,10pt)}]G.north east) --
  ([shift={(-15pt,10pt)}]G.north west) |-
  ([shift={(0pt,25pt)}]G.south west);
\draw[zeroarrow={0}{}]
  ([shift={(-25pt,-20pt)}]G.north west) -- ([shift={(0pt,-20pt)}]G.north west);
\draw[onearrow={1}{}]
  ([shift={(0pt,-20pt)}]G.north east) -- ([shift={(25pt,-20pt)}]G.north east);
\end{tikzpicture}}%
   \end{array}\]
~~~~~~~~~~~~~~~~~~~~~~~~~~~~~~~~~~\textbf{Superposition}
~~~~~~~~~~~~~~~~~~~~~~~~~~~~~~~~~~~~~~~~~~~~~~~~~~~~~~~~~~ \textbf{Yanking}
    \[  \begin{array}[c]{c}\resizebox{!}{1.75cm}{%
\begin{tikzpicture}
\node[myblock]
  (G) 
  {$f$};
      \node[mysquare,below=5pt of {G}] (S) {$g$}; 
\draw[rounded corners,onearrow={0.525}{}]
  ([shift={(0pt,20pt)}]G.south east) -|
  ([shift={(15pt,10pt)}]G.north east) --
  ([shift={(-15pt,10pt)}]G.north west) |-
  ([shift={(0pt,20pt)}]G.south west);
\draw[zeroarrow={0}{}]
  ([shift={(-25pt,-20pt)}]G.north west) -- ([shift={(0pt,-20pt)}]G.north west);
\draw[onearrow={1}{}]
  ([shift={(0pt,-20pt)}]G.north east) -- ([shift={(25pt,-20pt)}]G.north east);
  \draw[zeroarrow={0}{}]
  ([shift={(-30pt,-0.25cm)}]S.north west) -- ([shift={(0pt,-0.25cm)}]S.north west);
\draw[onearrow={1}{}]
  ([shift={(0pt,-0.25cm)}]S.north east) -- ([shift={(30pt,-0.25cm)}]S.north east);
      \draw[thick,dotted] ($(G.north west)+(-0.4,0.1)$)  rectangle ($(G.south east)+(0.4,-0.8)$);
\end{tikzpicture} }%
   \end{array}=
   \begin{array}[c]{c}\resizebox{!}{1.75cm}{%
\begin{tikzpicture}
\node[myblock]
  (G) 
  {$f$};
      \node[mysquare,below=5pt of {G}] (S) {$g$}; 
\draw[rounded corners,onearrow={0.525}{}]
  ([shift={(0pt,20pt)}]G.south east) -|
  ([shift={(15pt,10pt)}]G.north east) --
  ([shift={(-15pt,10pt)}]G.north west) |-
  ([shift={(0pt,20pt)}]G.south west);
\draw[zeroarrow={0}{}]
  ([shift={(-25pt,-20pt)}]G.north west) -- ([shift={(0pt,-20pt)}]G.north west);
\draw[onearrow={1}{}]
  ([shift={(0pt,-20pt)}]G.north east) -- ([shift={(25pt,-20pt)}]G.north east);
  \draw[zeroarrow={0}{}]
  ([shift={(-30pt,-0.25cm)}]S.north west) -- ([shift={(0pt,-0.25cm)}]S.north west);
\draw[onearrow={1}{}]
  ([shift={(0pt,-0.25cm)}]S.north east) -- ([shift={(30pt,-0.25cm)}]S.north east);
\end{tikzpicture} }%
   \end{array} ~~~~~~~~~~~~~ \begin{array}[c]{c}\resizebox{!}{1cm}{%
\begin{tikzpicture}
\node[]
  (G) 
  {};
  \node[below=5pt of G]
  (S) 
  {};
    \node[right=10pt of S]
  (T) 
  {};
    \node[left=10pt of S]
  (U) 
  {};
\draw[rounded corners,onearrow={0.525}{}]
  ([shift={(5pt,5pt)}]G.south east) -|
  ([shift={(25pt,10pt)}]G.north east) --
  ([shift={(-25pt,10pt)}]G.north west) |-
  ([shift={(-5pt,5pt)}]G.south west);
  \draw[zeroarrow={0}{}]
  ([shift={(0pt,-0.25cm)}]U.north west) -- ([shift={(12.65pt,10pt)}]S.north west);
\draw[zeroarrow={1}{}]
  ([shift={(-12.65pt,10pt)}]S.north east) -- ([shift={(0pt,-0.25cm)}]T.north east);
    \draw[onearrow={1}{}]
  ([shift={(0pt,-0.25cm)}]T.north east) -- ([shift={(20pt,-0.25cm)}]T.north east);
    \draw[zeroarrow={0}{}]
  ([shift={(-20pt,-0.25cm)}]U.north west) -- ([shift={(0pt,-0.25cm)}]U.north west);
  \end{tikzpicture} }%
   \end{array}=
   \begin{array}[c]{c}\resizebox{!}{0.25cm}{%
\begin{tikzpicture}
  \node[]
  (S) 
  {};
  \draw[zeroarrow={0}{}]
  ([shift={(-30pt,-0.25cm)}]S.north west) -- ([shift={(5pt,-0.25cm)}]S.north west);
\draw[onearrow={1}{}]
  ([shift={(-5pt,-0.25cm)}]S.north east) -- ([shift={(30pt,-0.25cm)}]S.north east);
  \end{tikzpicture} }%
   \end{array} \]
\end{definition}

\begin{definition}\label{addtdef} An \textbf{additive traced symmetric monoidal category} is a traced symmetric monoidal category which is an additive symmetric monoidal category such that the trace is a monoid morphism: 
 \[   \begin{array}[c]{c}\resizebox{!}{1.15cm}{%
\begin{tikzpicture}
\node[myblock]
  (G) 
  {$f+g$};
\draw[rounded corners,onearrow={0.525}{}]
  ([shift={(0pt,20pt)}]G.south east) -|
  ([shift={(15pt,10pt)}]G.north east) --
  ([shift={(-15pt,10pt)}]G.north west) |-
  ([shift={(0pt,20pt)}]G.south west);
\draw[zeroarrow={0}{}]
  ([shift={(-25pt,-20pt)}]G.north west) -- ([shift={(0pt,-20pt)}]G.north west);
\draw[onearrow={1}{}]
  ([shift={(0pt,-20pt)}]G.north east) -- ([shift={(25pt,-20pt)}]G.north east);
\end{tikzpicture}}%
\end{array} =   \begin{array}[c]{c}\resizebox{!}{1.15cm}{%
\begin{tikzpicture}
\node[myblock]
  (G) 
  {$f$};
\draw[rounded corners,onearrow={0.525}{}]
  ([shift={(0pt,20pt)}]G.south east) -|
  ([shift={(15pt,10pt)}]G.north east) --
  ([shift={(-15pt,10pt)}]G.north west) |-
  ([shift={(0pt,20pt)}]G.south west);
\draw[zeroarrow={0}{}]
  ([shift={(-25pt,-20pt)}]G.north west) -- ([shift={(0pt,-20pt)}]G.north west);
\draw[onearrow={1}{}]
  ([shift={(0pt,-20pt)}]G.north east) -- ([shift={(25pt,-20pt)}]G.north east);
\end{tikzpicture}}%
\end{array}\!+\!   \begin{array}[c]{c}\resizebox{!}{1.15cm}{%
\begin{tikzpicture}
\node[myblock]
  (G) 
  {$g$};
\draw[rounded corners,onearrow={0.525}{}]
  ([shift={(0pt,20pt)}]G.south east) -|
  ([shift={(15pt,10pt)}]G.north east) --
  ([shift={(-15pt,10pt)}]G.north west) |-
  ([shift={(0pt,20pt)}]G.south west);
\draw[zeroarrow={0}{}]
  ([shift={(-25pt,-20pt)}]G.north west) -- ([shift={(0pt,-20pt)}]G.north west);
\draw[onearrow={1}{}]
  ([shift={(0pt,-20pt)}]G.north east) -- ([shift={(25pt,-20pt)}]G.north east);
\end{tikzpicture} }%
\end{array}\quad \quad\begin{array}[c]{c}\resizebox{!}{1.15cm}{%
\begin{tikzpicture}
\node[myblock]
  (G) 
  {$0$};
\draw[rounded corners,onearrow={0.525}{}]
  ([shift={(0pt,20pt)}]G.south east) -|
  ([shift={(15pt,10pt)}]G.north east) --
  ([shift={(-15pt,10pt)}]G.north west) |-
  ([shift={(0pt,20pt)}]G.south west);
\draw[zeroarrow={0}{}]
  ([shift={(-25pt,-20pt)}]G.north west) -- ([shift={(0pt,-20pt)}]G.north west);
\draw[onearrow={1}{}]
  ([shift={(0pt,-20pt)}]G.north east) -- ([shift={(25pt,-20pt)}]G.north east);
\end{tikzpicture}}%
\end{array} =   \begin{array}[c]{c}\resizebox{!}{0.45cm}{%
\begin{tikzpicture}
      \node[mysquare] (T) {$0$}; 
\draw[zeroarrow={0}{}]
  ([shift={(-15pt,-0.25cm)}]T.north west) -- ([shift={(0pt,-0.25cm)}]T.north west);
  \draw[onearrow={1}{}]
  ([shift={(0pt,-0.25cm)}]T.north east) -- ([shift={(15pt,-0.25cm)}]T.north east);
\end{tikzpicture} }%
   \end{array} \] 
\end{definition}

\begin{example}\label{tracex} \normalfont It is well known that a compact closed category is canonically a traced symmetric monoi- \\ \noindent dal category, and conversely that every traced symmetric monoidal category induces a compact closed category (known as the $Int$-construction) \cite{abramsky1992new, hasegawa2009traced, joyal1996traced}. This relation still holds in the additive case. Any additive symmetric monoidal category which is compact closed is an additive traced symmetric monoidal category, and conversely given an additive traced symmetric monoidal category, its induced compact closed category is an additive symmetric monoidal category. This observation gives us the following well-known examples of additive traced symmetric monoidal categories (whose traces are induced from compact closed structure). 
\begin{enumerate}[{\em (i)}]
\item $\mathsf{REL}$ is an additive traced symmetric monoidal category where for a relation $R \subseteq (X \times Y) \times (X \times Z)$, its trace $\mathsf{Tr}^X(R) \subseteq Y \times Z$ is defined as $\mathsf{Tr}^X(R) := \lbrace (y,z) \vert~ \exists x\in X. ~ \left( (y,x), (x,z) \right) \in R \rbrace$. 
\item Let $\mathbb{K}$ be a field. Then $\mathsf{FVEC}_\mathbb{K}$ is an additive traced symmetric monoidal category where the trace is given by the standard partial trace of matrices/linear maps. In particular,  $\mathsf{FHilb}$ is an additive traced symmetric monoidal category. Explicitly, if $\lbrace x_1, \hdots, x_n \rbrace$ is a basis for a $\mathbb{K}$-vector space $X$, let $\lbrace x^*_1, \hdots, x^*_n \rbrace$ be the dual basis for the dual space $V^*$. Then for a linear map $f: X \otimes V \to X \otimes W$, its trace $\mathsf{Tr}^X(f): V \to W$ is defined as follows $\mathsf{Tr}^X(f)(v) := \sum \limits^{n}_{i=1} (x^*_i \otimes id_W)(f(x_i \otimes v))$. Note that the definition of trace is independent of the choice of basis. 
\end{enumerate}
\end{example}

We are now in a position to give the general statement as to why the only codifferential category structure on $\mathsf{FHilb}$ is the trivial one. To do so, we must assume additional assumptions on the endomorphisms of the monoidal unit. This hom-set is sometimes referred to as the set of scalars of a monoidal category, and in the case of an additive symmetric monoidal category, the internal semi-ring. 

\begin{theorem}\label{mainthm} Let $\mathbb{X}$ be an additive traced symmetric monoidal category with monoidal unit $K$ such that:
\begin{enumerate}[{\em (i)}]
\item $\mathbb{X}(K,K)$ is additively cancellative, that is, if $f + g = f+ h$ then $g=h$ for all $f,g,h \in \mathbb{X}(K,K)$; 
\item For any objects $A$ and $B$, $\mathsf{Tr}^A(1_A) \otimes \mathsf{Tr}^B(1_B) = 0$ if and only if $A$ or $B$ is a zero object.
\end{enumerate}
Then any codifferential category structure on $\mathbb{X}$ is trivial (Definition \ref{trivdef}).
\end{theorem} 
\begin{proof}Before starting the proof, we should first address these additional assumptions. Assumption (i) generalizes the need for additive inverses in the proof of Proposition \ref{ST}. On the other hand, recall that in a symmetric traced monoidal category, $\mathsf{Tr}^A(1_A)$ is regarded as the dimension of $A$. Then assumption (ii) is a sort of integral domain requirement that implies that if the dimension of $A \otimes B$ is $0$, then either $A$ or $B$ has dimension $0$. But assumption (ii) is slightly stronger by also requiring that if $A$ has dimension $0$ then $A$ is a zero object. This addresses the link between traces and dimensions in the proof of Proposition \ref{ST}. 

Now let $(\mathsf{S}, \mu, \eta, \mathsf{m}, \mathsf{u})$ be an algebra modality on $\mathbb{X}$ with deriving transformation $\mathsf{d}$. First define the natural transformation $\mathsf{d}^\circ$, known as the \textbf{coderiving transformation} \cite{cockettlemay2018}, as follows: 
\[   \begin{array}[c]{c}\resizebox{!}{1.25cm}{%
\begin{tikzpicture}
	\begin{pgfonlayer}{nodelayer}
		\node [style=none] (52) at (7.5, 1.75) {};
		\node [style=function1] (53) at (7.5, 1.75) {$\mathsf{m}$};
		\node [style=function2] (54) at (6, 1.25) {$\eta$};
		\node [style=none] (56) at (7, 1.5) {};
		\node [style=none] (57) at (7, 2) {};
		\node [style=object] (58) at (8.75, 1.75) {$\mathsf{S}A$};
		\node [style=object] (59) at (4.75, 1.75) {$:=$};
		\node [style=none] (63) at (2.75, 1.75) {};
		\node [style=function1] (64) at (2.75, 1.75) {$\mathsf{d}^\circ$};
		\node [style=object] (65) at (1.25, 2) {$\mathsf{S}A$};
		\node [style=object] (66) at (1.25, 1.5) {$A$};
		\node [style=none] (67) at (2.25, 2) {};
		\node [style=none] (68) at (2.25, 1.5) {};
		\node [style=object] (70) at (5.25, 2.25) {$\mathsf{S}A$};
		\node [style=object] (71) at (5.25, 1.25) {$A$};
		\node [style=object] (72) at (4.25, 1.75) {$\mathsf{S}A$};
		\node [style=none] (79) at (6.5, 2.25) {};
		\node [style=none] (80) at (6.5, 1.25) {};
	\end{pgfonlayer}
	\begin{pgfonlayer}{edgelayer}
		\draw [style=post] (53) to (58);
		\draw [style=wire] (65) to (67.center);
		\draw [style=wire] (66) to (68.center);
		\draw [style=wire] (71) to (54);
		\draw [style=wire] (54) to (80.center);
		\draw [style=wire] (80.center) to (56.center);
		\draw [style=wire] (79.center) to (57.center);
		\draw [style=post] (64) to (72);
		\draw [style=wire] (70) to (79.center);
	\end{pgfonlayer}
\end{tikzpicture}
  }%
\end{array}
\]
From a Fock space perspective, as explained by Fiore in \cite{fiore2015axiomatics}, the coderiving transformation is an annihilation operator \cite{geroch1985mathematical}. Using the linear rule and the product rule, it is easy to check that the deriving transformation and coderiving transformation satisfy the following identity \cite[Proposition 4.1]{cockettlemay2018}:
\[\begin{array}[c]{c}\resizebox{!}{0.85cm}{%
\begin{tikzpicture}
\node[myblock]
  (G) 
  {$\mathsf{d}$};
    \node[myblock,left=15pt of {G}]
  (H) 
  {$\mathsf{d}^\circ$};
  \draw[zeroarrow={0}{}]
  ([shift={(-15pt,-0.25cm)}]H.north west) -- ([shift={(0pt,-0.25cm)}]H.north west);
  \draw[zeroarrow={0}{}]
  ([shift={(-15pt,-0.75cm)}]H.north west) -- ([shift={(0pt,-0.75cm)}]H.north west);
\draw[zeroarrow={0}{}]
  ([shift={(-15pt,-0.5cm)}]G.north west) -- ([shift={(0pt,-0.5cm)}]G.north west);
\draw[onearrow={1}{}]
  ([shift={(0pt,-0.25cm)}]G.north east) -- ([shift={(25pt,-0.25cm)}]G.north east);
  \draw[onearrow={1}{}]
  ([shift={(0pt,-0.75cm)}]G.north east) -- ([shift={(25pt,-0.75cm)}]G.north east);
\end{tikzpicture}}%
\end{array} \!=\! \begin{array}[c]{c}\resizebox{!}{1.15cm}{%
\begin{tikzpicture}
\node[myblock]
  (G) 
  {$\mathsf{d}$};
    \node[myblock,right=25pt of {G}]
  (H) 
  {$\mathsf{d}^\circ$};
    \draw[zeroarrow={0}{}]
  ([shift={(-15pt,-1.25cm)}]G.north west) -- ([shift={(29pt,-1.25cm)}]G.north west);
  \draw[zeroarrow={0}{}]
  ([shift={(0,-0.75cm)}]G.north east) -- ([shift={(25pt,-1.25cm)}]G.north east);
  \draw[zeroarrow={0}{}]
  ([shift={(-25pt,-0.25cm)}]H.north west) -- ([shift={(0pt,-0.25cm)}]H.north west);
\draw[zeroarrow={0}{}]
  ([shift={(-15pt,-0.5cm)}]G.north west) -- ([shift={(0pt,-0.5cm)}]G.north west);
\draw[onearrow={1}{}]
  ([shift={(25pt,-1.25cm)}]G.north east) -- ([shift={(75pt,-1.25cm)}]G.north east);
  \draw[onearrow={1}{}]
  ([shift={(0pt,-0.5cm)}]H.north east) -- ([shift={(25pt,-0.5cm)}]H.north east);
    \draw[zeroarrow={0}{}]
  ([shift={(-25pt,-1.25cm)}]H.north west) -- ([shift={(0pt,-0.75cm)}]H.north west);
\end{tikzpicture}}%
\end{array} +     \begin{array}[c]{c}\resizebox{!}{0.75cm}{%
\begin{tikzpicture}
  \node[]
  (S) 
  {};
  \draw[zeroarrow={0}{}]
  ([shift={(-30pt,-0.25cm)}]S.north west) -- ([shift={(5pt,-0.25cm)}]S.north west);
\draw[onearrow={1}{}]
  ([shift={(-5pt,-0.25cm)}]S.north east) -- ([shift={(20pt,-0.25cm)}]S.north east);
    \draw[zeroarrow={0}{}]
  ([shift={(-30pt,0.25cm)}]S.north west) -- ([shift={(5pt,0.25cm)}]S.north west);
\draw[onearrow={1}{}]
  ([shift={(-5pt,0.25cm)}]S.north east) -- ([shift={(20pt,0.25cm)}]S.north east);
  \end{tikzpicture} }%
   \end{array} \] 
   This identity shows that $\mathsf{d}$ and $\mathsf{d}^\circ$ play the role of $x$ and $y$ in Proposition \ref{ST}. In fact, this above identity is well known relation between creation operators and annihilation operators in Fock spaces \cite{fiore2015axiomatics, geroch1985mathematical}. 
   
From the above identity, we obtain the following:
\begin{align*}
\begin{array}[c]{c}\resizebox{!}{1.25cm}{%
\begin{tikzpicture}
\node[myblock]
  (G) 
  {$\mathsf{d}$};
    \node[myblock,right=25pt of {G}]
  (H) 
  {$\mathsf{d}^\circ$};
  \draw[rounded corners,onearrow={0.525}{}]
  ([shift={(0pt,0.5cm)}]H.south east) -|
  ([shift={(15pt,10pt)}]H.north east) --
  ([shift={(-15pt,10pt)}]G.north west) |-
  ([shift={(0pt,0.5cm)}]G.south west);
  \draw[zeroarrow={0}{}]
  ([shift={(-25pt,-0.25cm)}]H.north west) -- ([shift={(0pt,-0.25cm)}]H.north west);
    \draw[zeroarrow={0}{}]
  ([shift={(0pt,-0.75cm)}]G.north east) -- ([shift={(0pt,-0.75cm)}]H.north west);
\end{tikzpicture}}%
\end{array} &= \begin{array}[c]{c}\resizebox{!}{1.5cm}{%
\begin{tikzpicture}
\node[myblock]
  (G) 
  {$\mathsf{d}$};
    \node[myblock,left=15pt of {G}]
  (H) 
  {$\mathsf{d}^\circ$};
    \draw[rounded corners,onearrow={0.525}{}]
  ([shift={(0pt,0.75cm)}]G.south east) -|
  ([shift={(15pt,10pt)}]G.north east) --
  ([shift={(-15pt,10pt)}]H.north west) |-
  ([shift={(0pt,0.75cm)}]H.south west);
      \draw[rounded corners,onearrow={0.525}{}]
  ([shift={(0pt,0.25cm)}]G.south east) -|
  ([shift={(25pt,20pt)}]G.north east) --
  ([shift={(-25pt,20pt)}]H.north west) |-
  ([shift={(0pt,0.25cm)}]H.south west);
\draw[zeroarrow={0}{}]
  ([shift={(-15pt,-0.5cm)}]G.north west) -- ([shift={(0pt,-0.5cm)}]G.north west);
\end{tikzpicture}}%
\end{array} \\
&= \begin{array}[c]{c}\resizebox{!}{1.75cm}{%
\begin{tikzpicture}
\node[myblock]
  (G) 
  {$\mathsf{d}$};
    \node[myblock,right=25pt of {G}]
  (H) 
  {$\mathsf{d}^\circ$};
    \draw[rounded corners,onearrow={0.525}{}]
  ([shift={(0pt,0.5cm)}]H.south east) -|
  ([shift={(15pt,10pt)}]H.north east) --
  ([shift={(-15pt,10pt)}]G.north west) |-
  ([shift={(0pt,0.5cm)}]G.south west);
        \draw[rounded corners,onearrow={0.525}{}]
  ([shift={(-29pt,-0.25cm)}]H.south east) -|
  ([shift={(25pt,20pt)}]H.north east) --
  ([shift={(-25pt,20pt)}]G.north west) |-
  ([shift={(29pt,-0.25cm)}]G.south west);
  \draw[zeroarrow={0}{}]
  ([shift={(0,-0.75cm)}]G.north east) -- ([shift={(25pt,-1.25cm)}]G.north east);
  \draw[zeroarrow={0}{}]
  ([shift={(-25pt,-0.25cm)}]H.north west) -- ([shift={(0pt,-0.25cm)}]H.north west);
    \draw[zeroarrow={0}{}]
  ([shift={(-25pt,-1.25cm)}]H.north west) -- ([shift={(0pt,-0.75cm)}]H.north west);
\end{tikzpicture}}%
\end{array} + \begin{array}[c]{c}\resizebox{!}{1.25cm}{%
\begin{tikzpicture}
\node[]
  (G) 
  {};
    \node[]
  (H) 
  {};
    \draw[rounded corners,onearrow={0.525}{}]
  ([shift={(-3pt,5pt)}]H.south east) -|
  ([shift={(15pt,10pt)}]H.north east) --
  ([shift={(-15pt,10pt)}]G.north west) |-
  ([shift={(5pt,5pt)}]G.south west);
        \draw[rounded corners,onearrow={0.525}{}]
  ([shift={(-3pt,-0.25cm)}]H.south east) -|
  ([shift={(25pt,20pt)}]H.north east) --
  ([shift={(-25pt,20pt)}]G.north west) |-
  ([shift={(5pt,-0.25cm)}]G.south west);
\end{tikzpicture}}%
\end{array} \\
&= \begin{array}[c]{c}\resizebox{!}{1.25cm}{%
\begin{tikzpicture}
\node[myblock]
  (G) 
  {$\mathsf{d}$};
    \node[myblock,right=25pt of {G}]
  (H) 
  {$\mathsf{d}^\circ$};
  \draw[rounded corners,onearrow={0.525}{}]
  ([shift={(0pt,0.5cm)}]H.south east) -|
  ([shift={(15pt,10pt)}]H.north east) --
  ([shift={(-15pt,10pt)}]G.north west) |-
  ([shift={(0pt,0.5cm)}]G.south west);
  \draw[zeroarrow={0}{}]
  ([shift={(-25pt,-0.25cm)}]H.north west) -- ([shift={(0pt,-0.25cm)}]H.north west);
    \draw[zeroarrow={0}{}]
  ([shift={(0pt,-0.75cm)}]G.north east) -- ([shift={(0pt,-0.75cm)}]H.north west);
\end{tikzpicture}}%
\end{array} + \begin{array}[c]{c}\resizebox{!}{0.75cm}{%
\begin{tikzpicture}
\node[left=-5pt of {H}]
  (G) 
  {};
    \node[]
  (H) 
  {};
    \draw[rounded corners,onearrow={0.525}{}]
  ([shift={(-3pt,5pt)}]H.south east) -|
  ([shift={(15pt,10pt)}]H.north east) --
  ([shift={(-15pt,10pt)}]H.north west) |-
  ([shift={(5pt,5pt)}]H.south west);
      \draw[rounded corners,onearrow={0.525}{}]
  ([shift={(-3pt,5pt)}]G.south east) -|
  ([shift={(15pt,10pt)}]G.north east) --
  ([shift={(-15pt,10pt)}]G.north west) |-
  ([shift={(5pt,5pt)}]G.south west);
\end{tikzpicture}}%
\end{array}
\end{align*} 
By assumption (i), $\mathbb{X}(K,K)$ is additively cancellative, so we obtain that ${0 = \mathsf{Tr}(1_{\mathsf{S}A}) \otimes \mathsf{Tr}(1_A)}$ for any object $A$. In particular for $\mathsf{S}A$, we have that $0 = \mathsf{Tr}(1_{\mathsf{S}\mathsf{S}A}) \otimes \mathsf{Tr}(1_{\mathsf{S}A})$. By assumption (ii), either $\mathsf{S}\mathsf{S}A$ is a zero object or $\mathsf{S}A$ is a zero object. If $\mathsf{S}A$ is a zero object we are done. So suppose that $\mathsf{S}\mathsf{S}A$ is a zero object. By the monad identities, $\mathsf{S}A$ is a retract of $\mathsf{S}\mathsf{S}A$. But a retract of a zero object is again a zero object. Therefore $\mathsf{S}A$ is a zero object. 
\end{proof} 

\begin{corollary}\label{maincor} Let $\mathbb{K}$ be a field of characteristic $0$. The only codifferential category structure on $\mathsf{FVEC}_\mathbb{K}$ is the trivial one. In particular, the only codifferential category structure on $\mathsf{FHilb}$ is the trivial one. 
\end{corollary}
\begin{proof} Let $\mathbb{K}$ be an arbitrary field. Then one has that $\mathsf{FVEC}_\mathbb{K}(\mathbb{K}, \mathbb{K}) \cong \mathbb{K}$, that is, linear maps from $\mathbb{K}$ to $\mathbb{K}$ correspond to precisely to elements of $\mathbb{K}$. Therefore, from now on, we will not distinguish between the two. In particular, assumption (i) of Theorem \ref{mainthm} holds since $\mathbb{K}$ is a field and so has additive inverses, which implies that $\mathbb{K}$ is additively cancellative. Regarding assumption (ii) of Theorem \ref{mainthm}, note that for any $r,s \in \mathbb{K}$, $r \otimes s = rs$. Therefore since $\mathbb{K}$ is a field, $\mathsf{Tr}^V(1_V) \otimes \mathsf{Tr}^W(1_W)=0$ if and only if $\mathsf{Tr}^V(1_V)=0$ or $\mathsf{Tr}^W(1_W)=0$. It is at this point that we require the additional assumption of having characteristic $0$. If $\mathbb{K}$ has characteristic $0$, then $V$ has dimension $n$ if and only if  $\mathsf{Tr}^V(1_V) = n$. Therefore, $V$ is a zero object if and only if $\mathsf{Tr}^V(1_V) = 0$. Putting everything together, so we conclude that in the case $\mathbb{K}$ has characteristic $0$, $\mathsf{FVEC}_\mathbb{K}$ is an additive traced symmetric monoidal category which satisfies assumption (i) and (ii) of Theorem \ref{mainthm}. 
\end{proof}

It is also worth noting that by self-duality of $\mathsf{FVEC}_\mathbb{K}$ (for a field $\mathbb{K}$ of characteristic $0$) and $\mathsf{FHilb}$, it also follows that the only differential category structure on either of them is the trivial one. In fact, under the same assumptions of Theorem \ref{mainthm}, one can also show that the only differential category structure for such a category is the trivial one. 

\begin{example}\label{failex} \normalfont In this paper we provided two non-trivial examples of codifferential categories with trace. Here we explain why each satisfies one of the assumptions of Theorem \ref{mainthm} but not the other, opening the door to non-trivial codifferential category structure. 
\begin{enumerate}[{\em (i)}]
\item Recall that in $\mathsf{REL}$ identity maps are the diagonal relations $\Delta_X = \lbrace (x,x) \vert~ x \in X \rbrace$ and the zero maps are the empty sets. In particular, $\mathsf{REL}(\lbrace \ast \rbrace, \lbrace \ast \rbrace) = \lbrace \emptyset, \Delta_{\lbrace \ast \rbrace} \rbrace$. Since addition is defined as the union of sets, $\mathsf{REL}(\lbrace \ast \rbrace, \lbrace \ast \rbrace)$ is not additively cancellative since in particular $\Delta_{\lbrace \ast \rbrace} + \Delta_{\lbrace \ast \rbrace} = \emptyset + \Delta_{\lbrace \ast \rbrace}$ but $\Delta_{\lbrace \ast \rbrace} \neq \emptyset$. Therefore $\mathsf{REL}$ fails assumption (i) of Theorem \ref{mainthm}. On the other hand $\emptyset$ is also the zero object in $\mathsf{REL}$, and one has that $\mathsf{Tr}^X(\Delta_X) = \emptyset$ if and only if $X = \emptyset$. From here, it is easy to check that $\mathsf{REL}$ satisfies assumption (ii) of Theorem \ref{mainthm}. 
\item As explained in the proof of Corollary \ref{maincor}, for any field $\mathbb{K}$,  $\mathsf{FVEC}_\mathbb{K}$ satisfies assumption (i) of Theorem \ref{mainthm}, and so in particular, $\mathsf{FVEC}_{\mathbb{Z}_2}$ does. On the other hand, $\mathbb{Z}_2$ has characteristic $2$ and so one can check that $\mathsf{Tr}^X(id_{\mathbb{Z}_2 \oplus \mathbb{Z}_2}) =0$. But $\mathbb{Z}_2 \oplus \mathbb{Z}_2$ is not a zero object, and so we can conclude that $\mathsf{FVEC}_{\mathbb{Z}_2}$ fails assumption (ii) of Theorem \ref{mainthm}.
\end{enumerate}
\end{example}

\section{Future Work}

In this paper, we showed that under the assumptions of Theorem \ref{mainthm}, the trace and the deriving transformation are incompatible, therefore providing necessary conditions for when a category only has trivial (co)differential category structure. In particular, this implies that $\mathsf{FHilb}$ is not an interesting (co)differential category. While this result may seem somewhat ``negative'', it does point in which direction to further study applications of differential categories to categorical quantum foundations. Here are some suggestions for potential directions and future work: 
\begin{enumerate}[{\em (i)}]
\item As mentioned in the introduction, there are interesting examples of non-trivial (co)differential categories with relations to quantum foundations in the literature, in particular relating to the Fock space. All these examples require ``infinite dimensional'' objects, and as we have shown that the ``finite dimensional'' case of $\mathsf{FHilb}$ is not interesting, it makes sense to look towards infinite dimensions. Potential interesting candidates for new (co)differential categories include the category of super-vector spaces and the category of all Hilbert spaces. 
\item The overall theory of differential categories can be split into four parts, each building on the previous part. The first part is differential categories, which provide the basic algebraic foundations of differentiation. The second part is Cartesian differential categories \cite{blute2009cartesian}, which axiomatizes differential calculus on Euclidean spaces. The third part is restriction differential categories \cite{cockett2011differential}, which studies differential calculus for partial functions. And the fourth part is tangent categories \cite{cockett2014differential}, which generalizes the theory of smooth manifolds and their tangent bundles. While $\mathsf{FHilb}$ may not be an interesting codifferential category, the category of finite-dimensional Hilbert spaces and smooth maps between them is a Cartesian differential category (and therefore also a tangent category). Hence, Cartesian differential categories, restriction differential categories, and tangent categories have many possible applications to categorical quantum foundations. 
\item We provided two examples of non-trivial trace/compact closed codifferential categories, both of which are quite interesting in their own way. $\mathsf{REL}$ is a fundamental example throughout theoretical computer science, in particular an important example in both categorical quantum foundations and the theory of differential categories. The codifferential category structure on $\mathsf{FVEC}_{\mathbb{Z}_2}$ has potential links to the ZW-calculus \cite{ZWcalc} and the CNOT category \cite{cockett2017category}. Therefore these two interesting examples should motivate the further study of traced/compact closed (co)differential categories. \end{enumerate}

\bibliographystyle{eptcs}
\bibliography{generic}
\end{document}